\RequirePackage{amssymb}
\documentclass[letterpaper, 10 pt, conference]{ieeeconf}
\IEEEoverridecommandlockouts
\onecolumn

\usepackage{color}
\usepackage{hyperref}
\usepackage{textcomp}
\usepackage{graphicx}
\usepackage{epsfig}
\usepackage{color,soul}
\usepackage{amsmath}
\usepackage{amsfonts}
\usepackage{mathrsfs}
\usepackage{xcolor}
\usepackage{tcolorbox}
\usepackage{mwe}
\usepackage{subcaption}
\usepackage{tabularx}
\usepackage[english]{babel}
\usepackage{theorem}
\theoremheaderfont{\itshape\bfseries}
{\theorembodyfont{\itshape}
\newtheorem{assumption}{\textbf{Assumption}}
\newtheorem{lemma}{\textbf{Lemma}}
\newtheorem{definition}{\textbf{Definition}}
\newtheorem{theorem}{\textbf{Theorem}}

\newtheorem{remark}{\textbf{Remark}}

\newtheorem{problem}{\textbf{Problem}}

\newcommand{\T}{^{\mbox{\tiny T}}}
\usepackage{url}
\newcommand{\eps}{\varepsilon}

\DeclareMathOperator{\re}{Re}

\DeclareMathOperator{\Span}{span}

\DeclareMathOperator{\round}{round}

\DeclareMathOperator{\im}{Im}

\newcommand{\R}{\mathbb{R}}

\newcommand{\N}{\mathcal{N}}
\let\leq\leqslant
\let\geq\geqslant

\newenvironment{proof}[1][Proof]%
{\par\noindent\textit{#1:\ }}%
{\hspace*{\fill} \rule{6pt}{6pt}}
\newenvironment{proof*}[1][Proof]%
{\par\noindent\textit{#1:\ }}{}

\IfFileExists{wasysym.sty}%
{
  
  \usepackage{wasysym}%
  }%
{}
\setcounter{MaxMatrixCols}{15}

\newenvironment{system}[1]%
{\setlength{\arraycolsep}{0.5mm}\left\{ \; \begin{array}{#1}}%
    {\end{array} \right.}
\newenvironment{system*}[1]%
{\setlength{\arraycolsep}{0.5mm} \begin{array}{#1}}%
  {\end{array}}

\usepackage{tikz}
\usetikzlibrary{shapes,calc,arrows,patterns,decorations.pathmorphing
  ,decorations.markings}
\usetikzlibrary{arrows.meta}

\begin{document}
	
	\title{Scalable $\delta$-Level Coherent State Synchronization of Multi-Agent Systems with Adaptive Protocols and Bounded Disturbances}
\author{Donya Nojavanzadeh, Zhenwei Liu, Ali Saberi and Anton A. Stoorvogel %
	\thanks{Donya Nojavanzadeh is with School of Electrical Engineering and Computer Science, Washington State University, Pullman, WA 99164, USA {\tt\small donya.nojavanzadeh@wsu.edu}}	
	\thanks{Zhenwei Liu is with College of Information Science and Engineering, Northeastern University, Shenyang 110819, P. R. China {\tt\small jzlzwsy@gmail.com}}	
	\thanks{Ali Saberi is with School of Electrical Engineering and Computer Science, Washington State University, Pullman, WA 99164, USA {\tt\small saberi@eecs.wsu.edu}} 	
	\thanks{Anton A. Stoorvogel is with Department of Electrical Engineering, Mathematics and Computer Science, University of Twente, P.O. Box 217, Enschede, The Netherlands {\tt\small A.A.Stoorvogel@utwente.nl}} 	
}

	\maketitle
	
	\begin{abstract}
In this paper, we study scalable $\delta-$level
  coherent state synchronization for multi-agent systems (MAS) where
  the agents are subject to bounded disturbances/noises. We propose a
  scale-free framework designed solely based on the knowledge of agent
  models and agnostic to the communication graph and the size of the network. We define the level of coherency for each agent as the norm
  of the weighted sum of the disagreement dynamics with its neighbors. The objective is to restrict the network's coherency level to $\delta$ without \emph{a-priori} information about the disturbance.
	\end{abstract}
	
	
	\footnotetext{ Z. Liu's work is supported by the Nature Science Foundation of China under 62273084 and the Nature Science Foundation of Liaonin Province, China under 2022-MS-111 and 2022JH25/10100008, and the Foundation of Key Laboratory of System Control and Information Processing, China under Scip202209.}

\section{Introduction}

Synchronization and consensus problems of MAS have become a hot topic
in recent years due to a wide range of applications in cooperative
control of MAS including robot networks, autonomous vehicles,
distributed sensor networks, and energy power systems. The objective
of synchronization of MAS is to secure an asymptotic agreement on a
common state or output trajectory by local interaction among agents,
see \cite{bai-arcak-wen,mesbahi-egerstedt,ren-book,wu-book,%
  liu-nojavanzedah-saberi-2022-book,saberi-stoorvogel-zhang-sannuti}
and references therein.

State synchronization inherently requires homogeneous networks. When
each agent has access to a linear combination of its own state
relative to that of the neighboring agents, it is called full-state
coupling. If the combination only includes a subset of the states
relative to the corresponding states of the neighboring agents, it is
called partial-state coupling. In the case of state synchronization
based on diffusive full-state coupling, the agent dynamics progress
from single- and double-integrator dynamics (e.g.\
\cite{saber-murray2,ren,ren-beard}) to more general dynamics (e.g.\
\cite{scardovi-sepulchre,tuna1,wieland-kim-allgower}). State
synchronization based on diffusive partial-state coupling has been
considered, including static design
\cite{liu-stoorvogel-saberi-nojavanzadeh-ijrnc19,liu-zhang-saberi-stoorvogel-auto},
dynamic design
\cite{kim-shim-back-seo,seo-back-kim-shim-iet,seo-shim-back,su-huang-tac,tuna3},
and designs based on localized information exchange with neighbors
\cite{chowdhury-khalil,scardovi-sepulchre}.  Recently, we have
introduced a new generation of \emph{scale-free} protocols (\emph{i.e.,} a protocol designed without using any knowledge about the communication graph and the size of the network) for
synchronization and almost synchronization of homogeneous and
heterogeneous MAS where the agents are subject to external
disturbances, input saturation, communication delays{,} and input
delays; see for example
\cite{liu-nojavanzadeh-saberi-saberi-stoorvogel-scl,%
  liu-donya-dmitri-saberi-stoorvogel-ccdc2020,liu-saberi-stoorvogel-nojavanzadeh-cdc19,%
  donya-liu-saberi-stoorvogel-ACC2020}.

Synchronization and almost synchronization in the presence of external
disturbances are studied in the literature, where three classes of
disturbances have been considered, namely
\begin{enumerate}
\item Disturbances and measurement noise with known frequencies
\item Deterministic disturbances with finite power
\item Stochastic disturbances with bounded variance
\end{enumerate}

For disturbances and measurement noise with known frequencies, it is
shown in \cite{zhang-saberi-stoorvogel-ACC2015} and
\cite{zhang-saberi-stoorvogel-CDC2016} that exact synchronization is
achievable. This is shown in \cite{zhang-saberi-stoorvogel-ACC2015}
for heterogeneous MAS with minimum-phase and non-introspective agents
and networks with time-varying directed communication graphs. Then,
\cite{zhang-saberi-stoorvogel-CDC2016} extended these results utilizing localized information exchange for
non-minimum phase agents.

For deterministic disturbances with finite power, the notion of
$H_\infty$ almost synchronization is introduced by Peymani et.al for
homogeneous MAS with non-introspective agents utilizing additional
communication exchange \cite{peymani-grip-saberi}. The goal of
$H_\infty$ almost synchronization is to reduce the impact of
disturbances on the synchronization error to an arbitrary degree
of accuracy (expressed in the $H_\infty$ norm). This work was extended
later in
\cite{peymani-grip-saberi-wang-fossen,zhang-saberi-grip-stoorvogel,%
  zhang-saberi-stoorvogel-sannuti2} to heterogeneous MAS with
non-introspective agents and without the additional communication exchange and
for network with time-varying graphs. $H_\infty$ almost
synchronization via static protocols is studied in
\cite{stoorvogel-saberi-liu-nojavanzadeh-ijrnc19} for MAS with passive
and passifiable agents. Recently, necessary and sufficient conditions
are provided in \cite{stoorvogel-saberi-zhang-liu-ejc} for solvability
of $H_\infty$ almost synchronization for homogeneous networks with
non-introspective agents and without additional communication
exchange. Finally, we developed a scale-free framework for $H_\infty$
almost state synchronization for homogeneous network
\cite{liu-saberi-stoorvogel-donya-almost-automatica} utilizing
suitably designed localized information exchange.

In the case of stochastic disturbances with bounded variance, the
concept of stochastic almost synchronization is introduced by
\cite{zhang-saberi-stoorvogel-stochastic2} where both stochastic
disturbance and disturbance with known frequency can be present. The
idea of stochastic almost synchronization is to make the stochastic
RMS norm of synchronization error arbitrarily small in the
presence of colored stochastic disturbances that can be modeled as the
output of linear time-invariant systems driven by white noise
with unit power spectral intensities. By augmenting this model with
the agent model one can essentially assume that stochastic
disturbances are white noise with unit power spectral intensities. In
this case, under linear protocols, the stochastic RMS norm of
synchronization error is the $H_2$ norm of the transfer function from
disturbance to the synchronization error. As such, one can formulate
the stochastic almost synchronization equivalently in a deterministic
framework where the objective is to make the $H_2$ norm of the
transfer function from disturbance to synchronization error
arbitrarily small. This deterministic approach is referred to as
almost $H_2$ synchronization problem which is equivalent to stochastic
almost synchronization problem. Recent work on $H_2$ almost
synchronization problem \cite{stoorvogel-saberi-zhang-liu-ejc} provides necessary and sufficient conditions for solvability of $H_\infty$ almost synchronization for homogeneous
networks with non-introspective agents and without additional
communication exchange. Finally, $H_2$ almost synchronization via
static protocols is studied in
\cite{stoorvogel-saberi-liu-nojavanzadeh-ijrnc19} for MAS with passive
and passifiable agents.

The above-described techniques of $H_\infty$ and $H_2$ almost synchronization of MAS have the following disadvantages.
\renewcommand\labelitemi{{\boldmath$\bullet$}}
\begin{itemize}
\item \textit{Tuning requirement:} The designed protocols for
  $H_\infty$ and $H_2$ almost synchronization of MAS are parameterized
  with a tuning parameter. The $H_\infty$ and $H_2$ norm of
  transfer function from the external disturbances can be made
  arbitrary small by a suitable choice of the tuning parameter but
  this relationship is strongly dependent on the communication graph.
\item \textit{Dependency on the size of disturbance:} In $H_\infty$
  and $H_2$ almost synchronization, the size of the synchronization error
  is linearly related to the size of the disturbance. Therefore, to guarantee a certain level of coherence, we need prior
  knowledge of the size of the disturbance to choose an
  appropriate tuning parameter $\delta$. 
\end{itemize}

On the other hand, in this paper, we consider scalable $\delta-$level
coherent state synchronization of homogeneous MAS in the presence of
bounded disturbances/noises, where one can achieve a guaranteed level
of coherence (\emph{i.e.,} the level of synchronization) without any prior
knowledge about the network and the
size of the disturbances. The contribution of this work is
twofold.
\begin{enumerate}
\item The protocols are designed solely based on the knowledge of the
  agent models, using no information about the communication
  network such as bounds on the spectrum of the associated Laplacian
  matrix or the number of agents. That is to say, the universal
  nonlinear protocols are scale-free and work for any
  communication network. There is not even the need to impose restrictions on the
  graph's connectivity.
\item We achieve scalable $\delta-$level coherent state
  synchronization for MAS in the presence of bounded
  disturbances/noises such that, for any given $\delta$, one can
  restrict the level of coherency of the network to $\delta$. The only
  assumption is that the disturbances are bounded (which is a very
  reasonable condition). But the protocol is independent of the bound
  and does not require any other knowledge about the disturbances.
\end{enumerate}

Note that we only consider disturbances to different agents and
not measurement noise. For further clarification see Remark \ref{rem2}.

\subsection*{Preliminaries on graph theory}

Given a matrix $M\in \mathbb{R}^{m\times n}$, $M\T$ denotes its
conjugate transpose. A square matrix $M$ is said to be Hurwitz stable
if all its eigenvalues are in the open left half complex plane.
$\im M$ denotes the image of matrix $M$, and $M_a\otimes M_b$ depicts the
Kronecker product of $M_a$ and $M_b$. $I_n$ denotes the
$n$-dimensional identity matrix and $0_n$ denotes $n\times n$ zero
matrix; sometimes we drop the subscript if the dimension is clear from
the context. For a signal $u$, we denote the $L_2$ norm by $||u||$ or
$||u||_2$, and the $L_\infty$ norm by $||u||_\infty$.

To describe the information flow among the agents we associate a
\emph{weighted graph} $\mathcal{G}$ to the communication network. The
weighted graph $\mathcal{G}$ is defined by a triple
$(\mathcal{V}, \mathcal{E}, \mathcal{A})$ where
$\mathcal{V}=\{1,\ldots, N\}$ is a node set, $\mathcal{E}$ is a set of
pairs of nodes indicating connections among nodes, and
$\mathcal{A}=[a_{ij}]\in \mathbb{R}^{N\times N}$ is the weighted
adjacency matrix with non-negative elements $a_{ij}$. Each pair in
$\mathcal{E}$ is called an \emph{edge}, where $a_{ij}>0$ denotes an
edge $(j,i)\in \mathcal{E}$ from node $j$ to node $i$ with weight
$a_{ij}$. Moreover, $a_{ij}=0$ if there is no edge from node $j$ to
node $i$. We assume there are no self-loops, \emph{i.e.,} we have
$a_{ii}=0$. A \emph{path} from node $i_1$ to $i_k$ is a sequence of
nodes $\{i_1,\ldots, i_k\}$ such that $(i_j, i_{j+1})\in \mathcal{E}$
for $j=1,\ldots, k-1$. A \emph{directed tree} is a sub-graph (a subset of
nodes and edges) in which every node has exactly one parent node
except for one node, called the \emph{root}, which has no parent
node. A \emph{directed spanning tree} is a sub-graph which is a
directed tree containing all the nodes of the original graph. If a
directed spanning tree exists, the root has a directed path to every
other node in the tree \cite{royle-godsil}.

For a weighted graph $\mathcal{G}$, the matrix
$L=[\ell_{ij}]$ with
\[
\ell_{ij}=
\begin{system}{cl}
\sum_{k=1}^{N} a_{ik}, & i=j,\\
-a_{ij}, & i\neq j,
\end{system}
\]
is called the \emph{Laplacian matrix} associated with the graph
$\mathcal{G}$. The Laplacian matrix $L$ has all its eigenvalues in the
closed right half plane and at least one eigenvalue at zero associated
with right eigenvector $\textbf{1}$ \cite{royle-godsil}. Moreover, if
the graph contains a directed spanning tree, the Laplacian matrix $L$
has a single eigenvalue at the origin and all other eigenvalues are
located in the open right-half complex plane \cite{ren-book}.

For a given graph $\mathcal{G}$ every maximal (by inclusion)
strongly connected sub-graph is called a bi-component of the graph. A
bi-component without any incoming edges is called a basic
bi-component. Every graph has at least one basic bi-component. Networks
have one unique basic bi-component if and only if the network contains
a directed spanning tree. In general, every node in a network can be
reached by at least one basis bi-component, see \cite[page
7]{stanoev-smilkov-2013}. 

In the absence of a directed spanning tree, the Laplacian matrix of
the graph has an eigenvalue at the origin with a multiplicity $k$
larger than $1$. This implies that it is a $k$-reducible matrix and
the graph has $k$ basic bi-components.  The book \cite[Definition
2.19]{wu-book} shows that, after a suitable permutation of the nodes,
a Laplacian matrix with $k$ basic bi-components can be written in the
following form:
\begin{equation}\label{Lstruc}
	L=\begin{pmatrix}
		L_0 & L_{01}     & L_{02} & \cdots  & L_{0k} \\
		0   & L_1        & 0      & \cdots  & 0 \\
		\vdots & \ddots  & \ddots & \ddots  & \vdots \\
		\vdots &         & \ddots & L_{k-1} & 0 \\
		0      & \cdots & \cdots  & 0       & L_k
	\end{pmatrix}
\end{equation}
where $L_1,\ldots, L_k$ are the Laplacian matrices associated to the
$k$ basic bi-components $\{ \mathcal{B}_1,\ldots, \mathcal{B}_k\}$ in the network. These matrices have a simple eigenvalue in $0$ because
they are associated with a strongly connected component. On the other
hand, $L_0$ contains all non-basic bi-components and is a grounded
Laplacian with all eigenvalues in the open right-half plane. After
all, if $L_0$ would be singular then the network would have an
additional basic bi-component.

\section{Problem formulation}

Consider a MAS consisting of $N$ identical linear
agents
\begin{equation}\label{agent-g-noise}
  \dot{x}_i=Ax_i+Bu_i+Ew_i,\quad i=1,\hdots, N
\end{equation}
where $x_i\in\mathbb{R}^n$, $u_i\in\mathbb{R}^m$, and
$w_i\in\mathbb{R}^w$ are state, input, and external disturbance/noise, respectively.

The communication network is such that each agent observes a weighted
combination of its own state relative to the state of other agents,
\emph{i.e.}, for the protocol of agent $i$, the signal
\begin{equation}\label{zeta1noise}
  \zeta_i=\sum_{j=1}^{N}a_{ij}(x_i-x_j)
\end{equation}
is available where $a_{ij}\geq0$ and $a_{ii}=0$. As explained before,
the matrix $\mathcal{A}=[a_{ij}]$ is the weighted adjacency matrix of a directed
graph $\mathcal{G}$. This matrix describes the communication
topology of the network, where the nodes of the network correspond to the
agents. We can also express the dynamics in terms of an associated
Laplacian matrix $L=[\ell_{ij}]_{N\times N}$, such that the signal
$\zeta_i$ in \eqref{zeta1noise} can be rewritten in the following form
\begin{equation}\label{zetanoise}
  \zeta_i= \sum_{j=1}^{N}\ell_{ij}x_j.
\end{equation}
The size of $\zeta_i(t)$ can be viewed as the level of coherency at agent $i$.

In this paper, we define the set of communication graphs as follows.

\begin{definition}\label{def1} 
  $\mathbb{G}^N$ denotes the set of directed graphs of $N$ agents.
\end{definition}

We make the following assumption. 
\begin{assumption}\label{ass}\mbox{}
  \begin{enumerate}
  \item $(A,B)$ is stabilizable.
  \item $\im E \subset \im B$.
  \item The disturbances $w_i$ are bounded for
    $i\in \{1,2,\ldots, N\}$. In other words, we have that
    $\| w \|_\infty < \infty$ for $i\in \{1,2,\ldots, N\}$.
  \end{enumerate}
\end{assumption}

\begin{remark}
  Assumptions 1.1 is an obvious necessary condition.  Assumption 1.2
  is a necessary condition for (almost) synchronization of a MAS with
  disturbances. Please see \cite[Lemma 1]{stoorvogel-saberi-zhang-liu-ejc}.
  As argued before, Assumption 1.3 is basically valid
  for any disturbance in the real world. The key point is that we do
  not need to know the bound.
\end{remark}

Next, in the following definition, we define the concept of
$\delta$-level-coherent state synchronization for the MAS with agents
\eqref{agent-g-noise} and communication information \eqref{zetanoise}.

\begin{definition}\label{delta-level}
  For any given $\delta>0$, the MAS \eqref{agent-g-noise} and
  \eqref{zetanoise} achieves $\delta$-level-coherent state
  synchronization if there exist a $T>0$ such that
  \[
    \|\zeta_i(t)\| \le\delta,
  \]
  for all $t>T$, for all $i\in\{1,\ldots, N\}$, and for any bounded disturbance.
\end{definition}

\begin{problem}\label{prob_x}
  Consider a MAS \eqref{agent-g-noise} with associated network
  communication \eqref{zetanoise} and a given parameter
  $\delta>0$. The \textbf{scalable {\boldmath $\delta$}-level-coherent
    state synchronization in the presence of bounded external
    disturbances/noises} is to find, if possible, a fully distributed
  nonlinear protocol for agent $i$ using only knowledge of agent
  models, \emph{i.e.,} $(A, B)$, and $\delta$, of the form
  \begin{equation}\label{out_dyn}
    \begin{system}{cl}
      \dot{x}_{i,c}&=f(x_{i,c},\zeta_i),\\
      u_i&=g(x_{i,c},\zeta_i),
    \end{system}
  \end{equation}
  where $x_{i,c}\in\mathbb{R}^{n_c}$ is the state of protocol, such
  that the MAS with the above protocol achieves
  $\delta$-level-coherent state synchronization in the presence of
  disturbances/noises. In other words, for any graph
  $\mathscr{G}\in\mathbb{G}$ with any size of the network $N$, and for
  all bounded disturbances $w_i$, the MAS achieves $\delta$-level
  coherent state synchronization as defined in Definition
  \ref{delta-level}.
\end{problem}

\section{Protocol design}

In this section, we design an adaptive protocol to achieve the
objectives of Problem \ref{prob_x} through three steps as following.
\begin{tcolorbox}[colback=white]
  \textbf{Step 1. Find matrix} $\mathbf{P}$\textbf{:} Under the assumption that $(A,B)$ is stabilizable, there exists a
  matrix $P>0$ satisfying the following algebraic Riccati equation
  \begin{equation}\label{eq-Riccati}
    A\T P + PA -PBB\T P + I =0.
  \end{equation}
  \textbf{Step 2. Choose parameter} $\mathbf{d}$\textbf{:} Define $\bar{\delta}=\delta^2\lambda_{\min}(P)$ and choose parameter $d$ such
  that
  \begin{equation}\label{dchoice}
    0<d< \bar{\delta}.
  \end{equation}
  Note that $\zeta_i\T P \zeta_i\le \bar{\delta}$ implies
  $\|\zeta_i(t)\| \le\delta$.
  
  \textbf{Step 3. Obtain protocol}\textbf{:}
  For any parameter $d$ satisfying \eqref{dchoice}, we design the
  following adaptive protocol
  \begin{equation}\label{protocol}
    \begin{system*}{ccl}
      \dot{\rho}_i &=& \begin{cases} \zeta_i\T PBB\T P \zeta_i &
        \text{ if } \zeta_i\T P \zeta_i \geq d,
        \\
        0 & \text{ if } \zeta_i\T P \zeta_i < d,
      \end{cases} \\
      u_i &=& -\rho_i B\T P \zeta_i.
    \end{system*}
  \end{equation}
\end{tcolorbox}

For classical adaptive controllers without disturbances, we would use
\[
  \dot{\rho}_i = \zeta_i\T PBB\T P \zeta_i
\]
and it can be shown that the scaling parameter $\rho_i$ remains finite
if no disturbances are present using classical techniques. However,
with persistent disturbances, this classical adaptation would imply
that the scaling parameter $\rho_i$ will become arbitrarily large over
time except for some degenerate cases. In this paper, we show that the
introduction of dead zone is a crucial modification that has very
desirable properties. We show the scaling parameters will remain
bounded. In the following, we will formally prove the key characteristics of this protocol. Unfortunately, the introduction of
the deadzone makes the proofs quite tricky even though the simulations
will illustrate very nice behavior.

We have the following theorem.

\begin{theorem}\label{theorem}
  Consider a MAS \eqref{agent-g-noise} with associated network
  communication \eqref{zetanoise} and a given parameter $\delta>0$
  satisfying Assumption \ref{ass}. The \textbf{scalable {\boldmath
      $\delta$}-level-coherent state synchronization in the presence
    of bounded external disturbances/noises} as stated in problem
  \ref{prob_x} is solvable. In particular, protocol \eqref{protocol}
  with any $d$ satisfying \eqref{dchoice} solves
  $\delta$-level-coherent state synchronization in the presence of
  disturbances/noises $w_i$, for an arbitrary number of agents $N$ and
  for any graph $\mathscr{G}\in\mathbb{G}^N$.
\end{theorem}

\begin{remark}\label{rem2}
  It is easy to show that measurement noise that converges to zero
  asymptotically will not affect synchronization since, eventually, it
  will be arbitrarily small.  However, if we have arbitrary
  measurement noise that is bounded, for instance, $|v_i(t)|<V$, then
  it is very obvious that one can never guarantee synchronization with
  an accuracy of less than $2V$. Bounded measurement noise (with upper
  bound $V$) in our argument will impose a lower bound on $d$ (and
  hence on $\delta$) to avoid that in the worst case the scaling
  parameter will converge to infinity.  This will create completely
  different arguments and results. In particular, we aimed to
  find protocols that do not depend on a specific upper bound for the
  disturbances. With measurement noise, both the choice of $d$ in the
  protocol as well as the accuracy of our synchronization will depend
  on the upper bound $V$.
\end{remark}

\begin{proof}
  The Laplacian matrix of the system in general has the form
  \eqref{Lstruc}. We note that if we look at the dynamics of the
  agents belonging to one of the basic bi-components then these
  dynamics are not influenced by the other agents and hence can be
  analyzed independent of the rest of network. The network within one
  of the basic bi-components is strongly connected and we can apply
  Lemmas \ref{lem2a} and \ref{lem2} to guarantee that the $\rho_i$
  associated to basic bi-components are all bounded.

  Next, we look at the full network again. We have already established
  that the $\rho_i$ associated to basic bi-components are all
  bounded. Then we can again apply Lemmas \ref{lem2a} and \ref{lem2}
  to conclude that the other $\rho_i$ not associated to basic
  bi-components are also bounded.

  After having established that all $\rho_i$ are bounded, we
  can then apply Lemma \ref{lem1} to conclude that we achieve scalable
  $\delta$-level-coherent state synchronization.
\end{proof}

\begin{lemma}\label{lem1}
  Consider a number of agents $N$ and a graph
  $\mathscr{G}\in\mathbb{G}^N$. Consider MAS \eqref{agent-g-noise} with
  associated network communication \eqref{zetanoise} and a given
  parameter $\delta>0$. Assume Assumption \ref{ass} is satisfied.
  Choose any $d$ satisfying \eqref{dchoice}. If all $\rho_i$ remain
  bounded then there exists a $T>0$ such that
  \begin{equation}\label{probdef}
    \zeta_i\T(t) P \zeta_i(t) \leq \bar{\delta},
  \end{equation}
  for all $t>T$ and for all $i=1,\ldots, N$.
\end{lemma}

\begin{proof}[Proof of Lemma \ref{lem1}]
  We define
  \[
    x=\begin{pmatrix} x_1 \\ \vdots \\ x_N \end{pmatrix},\quad
    w=\begin{pmatrix} w_1 \\ \vdots \\ w_N \end{pmatrix},
  \]
  and
  \[
    \rho=\begin{pmatrix}
      \rho_1 & 0      & \cdots & 0 \\
      0      & \rho_2 & \ddots & \vdots \\
      \vdots & \ddots & \ddots & 0 \\
      0      & \cdots & 0      & \rho_N
    \end{pmatrix}.
  \]
  We obtain
  \begin{equation}\label{syscomp}
    \dot{x}=(I\otimes A)x-(\rho L\otimes BB\T P)x + (I\otimes E)w.
  \end{equation}
  Next, we use
  \[
    \zeta_i = (L_i \otimes I)x,
  \]
  where $L_i$ is the $i^{\text{th}}$ row of $L$ for $i=1,\ldots, N$.
  We obtain
  \begin{equation}\label{barxt3}
    \zeta = (L \otimes I) x,
  \end{equation}
  where
  \[
    \zeta=\begin{pmatrix} \zeta_1 \\ \vdots \\ \zeta_N \end{pmatrix}.
  \]
  By using \eqref{syscomp}, we obtain
  \begin{equation}\label{barxt4}
    \dot{\zeta} = (I\otimes A)\zeta
    -[L\rho \otimes BB\T
    P]\zeta + (L\otimes E)w. 
  \end{equation}
  Convergence of $\rho_i$ for $i=1,\ldots, N$ implies that for any
  $\eps>0$ there exists $T>T_0$ such that
  \begin{equation}\label{rhoeps}
    \rho_i(t_2)-\rho_i(t_1) < \eps,
  \end{equation}
  for all $t_2>t_1>T$, and for all $i=1,\ldots, N$. This implies
  \[
    B\T P \zeta_i = p_i+q_i
  \]
  with
  \begin{equation}\label{epsM}
    \int_T^\infty p_i\T p_i \textrm{d}t <\eps,\qquad \| q_i \| \leq M,
  \end{equation}
  where $p_i(t)=B\T P\zeta_i$ and $q_i(t)=0$ if $\zeta_i\T(t)P\zeta_i(t)>d$ and
  $p_i(t)=0$ and $q_i(t)=B\T P\zeta_i$ if
  $\zeta_i\T(t)P\zeta_i(t)>d$. The bounds in \eqref{epsM} then follow
  from \eqref{rhoeps} and our protocol \eqref{protocol}.
  
  For any $i\in \{1,\ldots, N\}$, from \eqref{barxt4} we have that
  \begin{equation*}
    \dot{\zeta}_i = A\zeta_i
    -[L_i\rho \otimes BB\T P]\zeta + (L_i\otimes E)w. 
  \end{equation*} 
  Define
  \[
    V_i = \zeta_i\T P \zeta_i,
  \]
  then we obtain
  \begin{equation}\label{barxt7}
    \dot{V}_i = -\zeta_i\T \zeta_i + \zeta_i\T PBB\T P \zeta_i
    -2\rho_i \zeta_i\T \left[ L_{i} \otimes PBB\T P
    \right] \zeta+ 2\zeta_i\T (L_i \otimes PE) w. 
  \end{equation}
  Since $w$ is bounded and we have the bound \eqref{epsM}, we
  find that there exists a constant $R$ such that $\| s_i(t) \| < R$
  for all $t>T_0$, where
  \begin{equation}\label{si}
    s_i =-2(L_i \rho \otimes I)\begin{pmatrix} q_1 \\ \vdots \\
        q_N \end{pmatrix} +2(L_i \otimes X) w.
  \end{equation}
  Note that in \eqref{si} we used Assumption \ref{ass} implying that
  there exists a matrix $X$ such that $E=BX$. We also define
  \[
    r_i = B\T P \zeta_i,\qquad v_i = -2 \rho_i(L_i\otimes
    I) \begin{pmatrix} p_1 \\ \vdots \\ p_N \end{pmatrix}.
  \]
  Note that there exists $M_2$ such that
  \begin{equation}\label{M22}
    \int_T^\infty v_i\T (t) v_i(t) \textrm{d}t \leq M_2^2  \eps
  \end{equation}
  using \eqref{epsM}.
  We obtain from \eqref{barxt7} that
  \begin{equation}\label{barxt8}
    \dot{V}_i \leq -\eta V_i + r_i\T r_i + r_i\T v_i + r_i\T s_i,
  \end{equation}
  where $\eta= \lambda_{\max}(P)^{-1}$. Choose $\eps>0$ such that
  \begin{equation}\label{eps}
    R\sqrt{\eps}\leq \eta d,\qquad
    (M_2+1)\eps+\sqrt{\eps}R < \bar{\delta} -d.
  \end{equation}
 Next, assume $V_i >d$
  in the interval $[t_1,t_2]$ with $t_1>t_2>T$. In that case
  \[
    \int_{t_1}^{t_2}\, r_i\T(t) r_i(t)\, \textrm{d}t  =
    \rho_i(t_2)-\rho_i(t_1) < \eps 
  \]
  which implies
  \begin{equation}\label{barxt9}
    V_i(t_2)-V_i(t_1) \leq -\eta d(t_2-t_1) + (M_2+1)\eps +
    \sqrt{\eps}R \sqrt{t_2-t_1},
  \end{equation}
  where we used \eqref{M22}, \eqref{barxt8} and the fact that
  \begin{align*}
    \int_{t_1}^{t_2}\, r_i\T(t) s_i(t)\, \textrm{d}t
    &\leq \left( \int_{t_1}^{t_2}\, r_i\T(t) r_i(t)\, \textrm{d}t\right)^{1/2}
    \left( \int_{t_1}^{t_2}\, s_i\T(t) s_i(t)\, \textrm{d}t
    \right)^{1/2},\\
    \int_{t_1}^{t_2}\, r_i\T(t) v_i(t)\, \textrm{d}t
    &\leq \left( \int_{t_1}^{t_2}\, r_i\T(t) r_i(t)\, \textrm{d}t\right)^{1/2}
    \left( \int_{t_1}^{t_2}\, v_i\T(t) v_i(t)\, \textrm{d}t
    \right)^{1/2}.
  \end{align*}
  From \eqref{barxt9}, it is clear that if we increase $t_2$ (and keep
  $t_1$ fixed) then there will be a moment that $V_i(t_2)$ becomes
  negative which yields a contradiction. Therefore, there exists some
  $T_1>t_1>T$ for which $V_i(T_1)\leq d$. Next, we prove that for all
  $t>T_1$ we will have that $V_i(t)<\bar{\delta}$. We will show this
  by contradiction. If $ V_i(t)\geq \bar{\delta}$ for some $t>T_1$
  while $V_i(T_1)<d$, then there must exist $t_4>t_3>T_1$ such that
  \[
    V_i(t_3)=d \text{ and } V_i(t_4)=\bar{\delta} \text{ while }
    V_i(t)\geq d \text{ for } 
    t\in [t_3,t_4].
  \]
  Similar to \eqref{barxt9}, we get
  \[
    V_i(t_4)-V_i(t_3) \leq -\eta d(t_4-t_3) + (M_2+1)\eps + \sqrt{\eps}R
    \sqrt{t_4-t_3}.
  \]
  If $t_4-t_3>1$, then this implies
  \begin{equation*}
    \bar{\delta} -d = V_i(t_4)-V_i(t_3) \leq -\eta d(t_4-t_3) + (M_2+1)\eps
    +\sqrt{\eps}R (t_4-t_3) \leq (M_2+1)\eps < \bar{\delta} -d 
  \end{equation*}
  using \eqref{eps} which yields a contradiction. On the other hand,
  if $t_4-t_3\leq 1$ then we obtain
  \begin{equation*}
    \bar{\delta} -d = V_i(t_4)-V_i(t_3) \leq  (M_2+1)\eps +
    \sqrt{\eps}R \sqrt{t_4-t_3} \leq (M_2+1)\eps+\sqrt{\eps} R < \bar{\delta} -d
  \end{equation*}
  which also yields a contradiction. In this way, we can show for any
  $i\in \{ 1,\ldots,N\}$, \eqref{probdef} is satisfied for sufficiently large $t$.
\end{proof}   

In the first Lemma, we showed that if all the adaptive parameters
remain bounded then we obtain our desired result. We first establish
that these parameters can at most grow linearly over time. This is
then used in a later lemma to prove that the adaptive parameters are
actually bounded and we can use Lemma \ref{lem1} to establish synchronization.

\begin{lemma}\label{lem2a}
  Consider MAS \eqref{agent-g-noise} with associated network
  communication \eqref{zetanoise} and protocol \eqref{protocol}.
  Assume Assumption \ref{ass} is satisfied. Additionally, assume that
  either the $\rho_i$ associated to agents belonging to the basic
  bi-components are bounded or the graph is strongly
  connected. In that case, there exists $\eta,\mu >0$ such that
  \[
    \rho_i (t_2) -\rho_i(t_1) < \eta (t_2-t_1) + \mu,
  \]
  for any $t_1,t_2$ with $t_2>t_1$, and for $i=1,\ldots,N$. 
\end{lemma}

\begin{proof}  
  Using the notation of Lemma \ref{lem1} we obtain \eqref{syscomp}.  We
  first consider the case that $\rho_i$ is unbounded for $i\leq k$
  while $\rho_i$ is bounded for $i>k$ with $k<N$. We have
  \begin{equation*}
    L = \begin{pmatrix}
      L_{11} & L_{12} \\ L_{21} & L_{22} 
    \end{pmatrix},\quad
    x^k =\begin{pmatrix} x_1 \\ \vdots \\ x_k \end{pmatrix}, \quad
    x^k_c =\begin{pmatrix} x_{k+1} \\ \vdots \\
      x_{N} \end{pmatrix},\quad
    \zeta^k =\begin{pmatrix} \zeta_{1} \\ \vdots \\
      \zeta_{k} \end{pmatrix},\quad
    \zeta^k_c =\begin{pmatrix} \zeta_{k+1} \\ \vdots \\
      \zeta_{N} \end{pmatrix},
  \end{equation*}
  with $L_{11}\in \R^{k\times k}$. If all the agents associated to
  basic bi-components have a bounded $\rho_i$ this implies that agents
  associated to $i=1,\ldots,k$ are not associated to basic
  bi-components which implies that $L_{11}$ is invertible. On the other
  hand, if the network is strongly connected we always have that $L_{11}$
  is invertible since $k<N$. There exist $s^k_c\in L_\infty$ and
  $v^k_c\in L_2$ with
  \begin{equation}\label{K1K2}
    \| s^k_c \|_\infty < K_1,\quad \| v^k_c \|_2 < K_2
  \end{equation}
  for suitable chosen $K_1$ and $K_2$ such that
  \[
    (I \otimes B\T P) \zeta^k_c = \begin{pmatrix} B\T P \zeta_{k+1} \\ 
      \vdots \\ B\T P \zeta_{N} \end{pmatrix} =
    \begin{pmatrix} v_{k+1} \\ \vdots \\ v_{N} \end{pmatrix} +
    \begin{pmatrix} s_{k+1} \\ \vdots \\ \bar{s}_{N-k} \end{pmatrix} 
    = v^k_c+s^k_c.
  \]
  This is easily achieved by setting $v_i(t)=0$
  and $s_i(t)=B\T P \zeta_i(t)$ if
  $\zeta_i\T(t) P \zeta_i(t) <d$, while for
  $\zeta_i\T(t) P \zeta_i(t) \geq d$ we set
  $v_i(t)=B\T P \zeta_i(t)$ and $s_i(t)=0$.  It is obvious
  that this construction yields that $s_i\in L_\infty$ while the
  fact that the $\rho_i$ are bounded for $i=k+1,\ldots, N$ implies that
  $v_i\in L_2$ (note that $\dot{\rho}_i = v_i\T v_i$
  in this construction). We define
  \[
    \hat{x}^k= x^k+(L_{11}^{-1}L_{12}\otimes I) x^k_c.
  \] 
  Using \eqref{syscomp}, we then obtain
  \begin{equation}\label{barxt2xxx}
    \dot{\hat{x}}^k = (I\otimes A)\hat{x}^k
    -[\rho^{k} L_{11} \otimes BB\T
    P]\hat{x}^k  -\left[L_{11}^{-1}L_{12}\rho^k_c
      \otimes B\right](s^k_c+v^k_c)+
    [{\setlength{\arraycolsep}{2mm}\begin{pmatrix} I 
      & L_{11}^{-1}L_{12} \end{pmatrix}} \otimes BX]w 
  \end{equation}
  where we used, as before, that $E=BX$ while
  \[
    \rho^k=\begin{pmatrix}
      \rho_1 & 0      & \cdots & 0 \\
      0      & \rho_2 & \ddots & \vdots \\
      \vdots & \ddots & \ddots & 0 \\
      0      & \cdots & 0      & \rho_k
    \end{pmatrix},\quad
    \rho^k_c=\begin{pmatrix}
      \rho_{k+1} & 0      & \cdots & 0 \\
      0      & \rho_{k+2} & \ddots & \vdots \\
      \vdots & \ddots & \ddots & 0 \\
      0      & \cdots & 0      & \rho_N
    \end{pmatrix}.
  \]
  Define
  \begin{align*}
    \hat{s}^k &=-(L_{11}^{-1}L_{12}\rho^k_c \otimes I)s^k_c
    +\left[{\setlength{\arraycolsep}{2mm}\begin{pmatrix} I &
        L_{11}^{-1}L_{12} \end{pmatrix}} \otimes X\right]w,\\ 
    \hat{v}^k &=-(L_{11}^{-1}L_{12}\rho^k_c \otimes I)v^k_c , 
  \end{align*}
  then \eqref{K1K2} in combination with the boundedness of
  $\rho^k_c$ implies that there exists $K_3$ and
  $K_4$ such that
  \begin{equation}\label{K3K4}
    \| \hat{s}^k \|_\infty < K_3,\quad \| \hat{v}^k \|_2 < K_4.
  \end{equation}
  We obtain
  \begin{equation}\label{barxt2xxxy}
    \dot{\hat{x}}^k = (I\otimes A)\hat{x}^k
    -[\rho^k L_{11} \otimes BB\T
    P]\hat{x}^k+[I \otimes B](\hat{s}^k+\hat{v}^k)
  \end{equation}
  and we define
  \begin{equation}\label{Vj}
    V_k= (\hat{x}^k)\T (\rho^{-k} H^k \otimes P) \hat{x}^k,
  \end{equation}
  with $\rho^{-k}=(\rho^k)^{-1}$
  while
  \begin{equation}\label{HN}
    H^k=\begin{pmatrix}
      \alpha_1 & 0      & \cdots & 0 \\
      0      & \alpha_2 & \ddots & \vdots \\
      \vdots & \ddots & \ddots & 0 \\
      0      & \cdots & 0      & \alpha_k
    \end{pmatrix}.
  \end{equation}
  By \cite[Theorem 4.25]{qu-book-2009} we can choose
  $\alpha_1,\ldots,\alpha_k>0$ such that $H^k L +L\T H^k > 0$. It is
  easily seen that this implies that there exists a $\gamma$ such that
  \begin{equation}\label{HkL11}
    H^kL_{11}+L_{11}\T H^k > 3\gamma L_{11}\T L_{11}.
  \end{equation}
  We get from \eqref{barxt2xxxy} that
  \[
    \dot{V}_k \leq (\hat{x}^k)\T \left[ \rho^{-k} H^k \otimes (-I+PBB\T P) \right]
    \hat{x}^k-(\hat{x}^k)\T \left[
      (H^kL_{11}+L_{11}\T H^k)
      \otimes PBB\T P \right] \hat{x}^k \\
    + 2(\hat{x}^k)\T \left[ \rho^{-k} H^k \otimes PB\right] (\hat{s}^k+\hat{v}^k),
  \]
  where we used that $V_k$ is decreasing
  in $\rho_i$ for $i=1,\ldots, k$. The above yields for $t>T$ that
  \begin{equation*}
    \dot{V}_k \leq -V_k 
    - 2\gamma (\hat{x}^k)\T \left[ L_{11}\T L_{11} \otimes PBB\T P \right] \hat{x}^k \\
    + 2(\hat{x}^k)\T \left[\rho^{-k} H^k \otimes PB\right] (\hat{s}^k+\hat{v}^k),
  \end{equation*}
  provided $T$ is such that
  \[
    \rho^{-k} H^k  < \gamma I,
  \]
  for $t>T$ which is possible since we have
  $\rho_i\rightarrow \infty$ for $i=1,\ldots, k$.  Note that given \eqref{K3K4}
  there exists some fixed $\alpha$ such that
  \begin{equation}\label{lastone4xx}
    \int_{T}^\infty\, \| \check{s}^k(t) \|^2\, \textrm{d}t   \leq \alpha,\qquad
    \sup_{t\in[T,\infty)} \| \check{v}^k(t) \| \leq \alpha,
  \end{equation}
  with
  \begin{align*}
    \check{s}^k &= \tfrac{1}{\sqrt{\gamma}}\left[ L_{11}^{-1}
      \rho^{-k}H^k \otimes I\right] \hat{s}^k,\\ 
    \check{v}^k &= \tfrac{1}{\sqrt{\gamma}}\left[ L_{11}^{-1}
      \rho^{-k}H^k \otimes I\right]\hat{v}^k. 
  \end{align*}
  We get
  \begin{equation}\label{lastone2xx}
    \dot{V}_k \leq - V_k 
    - \gamma (\hat{x}^k)\T \left[ L_{11}\T L_{11} \otimes PBB\T P \right] \hat{x}^k 
    + (\check{s}^k)\T \check{s}^k + (\check{v}^k)\T\check{v}^k,
  \end{equation}
  for $t>T$. Moreover, 
  \begin{equation}\label{rfgt}
     (\hat{x}^k)\T \left[ L_{11}\T L_{11} \otimes PBB\T
      P \right] \hat{x}^k \geq  \sum_{i=1}^k \dot{\rho}_i,
  \end{equation}
  since $(L_{11}\otimes I)\hat{x}^k=\zeta^k$. Hence \eqref{lastone2xx} implies
  \begin{equation}\label{ttg1}
    \dot{V}_k \leq -V_k -\gamma
    \sum_{i=1}^k \dot{\rho}_{i} + (\check{s}^k)\T \check{s}^k +
    (\check{v}^k)\T\check{v}^k.  
  \end{equation} 
  Note that the bounds in \eqref{lastone4xx} combined with the
  inequality \eqref{ttg1} for $t>T$ implies that there exists some
  $\eta,\mu>0$ such that
  \begin{equation}\label{ttg4}
    \tilde{\rho}_i(t_2)-\tilde{\rho}_i(t_1) <\eta (t_2-t_1) + \mu,
  \end{equation}
  for $i=1,\ldots, k$ and all $t_2,t_1>T$. Clearly, if
  $\rho_{k+1},\ldots,\rho_N$ are all bounded we trivially obtain \eqref{ttg4} for
  $i=k+1,\ldots,N$. 

  If \eqref{ttg4} is satisfied for $t_1,t_2>T$, then clearly the lemma immediately
  follows for $t_1,t_2>T$. Since for $t<T$ all signals are bounded it
  is clear that for $t<T$ the $\rho_i$ can grow at most linearly and
  hence we can obtain the result for all $t_1,t_2>0$.

  Next, we consider the case that all $\rho_i$ are unbounded. In this
  case, we assumed the graph is strongly connected and hence by Lemma
  \ref{2.8} presented in the appendix there exists $\alpha_1,\ldots,\alpha_N>0$ such that
  \eqref{Hlyap} is satisfied with $H^N$ given by \eqref{HN} for $k=N$.
  We define
  \begin{equation}\label{VN}
    V_N= x\T \left[ Q_{\rho} \otimes P 
    \right] x,
  \end{equation}
  where
  \begin{equation}\label{Qrho}
    Q_{\rho} = \rho^{-N} \left( H^N\rho^N - \mu_N
      \textbf{h}_N\textbf{h}_N\T \right) \rho^{-N}
  \end{equation}
  with $\rho^{-N}=(\rho^N)^{-1}$, while
  \[
    \mu_N=\frac{1}{\sum_{i=1}^N \alpha_i\rho_i^{-1}},\qquad
    \textbf{h}_N = \begin{pmatrix} \alpha_1 \\ \vdots \\ \alpha_N \end{pmatrix}.
  \]
  From Lemma \ref{2.9} in the appendix, we know that $Q_{\rho}$ is
  decreasing in $\rho_i$ for $i=1,\ldots N$.  Note that
  $Q_\rho \rho L=H^N L$.  We get from \eqref{syscomp} that
  \begin{equation}\label{eqref}
    \dot{V}_N \leq x\T \left[ Q_\rho \otimes (-I+PBB\T P) \right] x
    -x\T \left[ (H^NL+L\T H^N) \otimes PBB\T P \right] x 
    + 2x\T \left[ Q_{\rho} \otimes PB\right] Xw,
  \end{equation}
  where we use again that there exists $X$ such that $E=BX$. It is
  easily verified that $Q_{\rho}\textbf{1}=0$. Moreover
  $\text{Ker} L = \Span \{ \textbf{1} \}$ the network is strongly
  connected. Therefore
  \[
    \text{Ker} Q_{\rho} \subset \text{Ker} L\T L.
  \]
  Together with the fact that $\rho_j\rightarrow \infty$ for
  $j=1,\ldots, N$ and therefore $Q_{\rho}\rightarrow 0$ this implies
  that there exists $T$ such that
  \begin{equation}\label{Qrhobound}
    Q_{\rho} < \gamma L\T L,
  \end{equation}
  is satisfied for $t>T$. 

  The above together with \eqref{Hlyap} yields for $t>T$ that
  \begin{equation*}
    \dot{V}_N \leq -V_N
    - 2\gamma x\T \left[ L\T L \otimes PBB\T P \right] x \\
    + 2 x\T \left[ Q_{\rho} \otimes PBX\right] w.
  \end{equation*}
  Note that
  \[
    (I-\tfrac{1}{N} \textbf{1}\textbf{1}\T)Q_\rho = Q_\rho.
  \]
  Choose $\nu$ such that
  \[
    (I-\tfrac{1}{N} \textbf{1}\textbf{1}\T) \leq \nu \gamma L\T L,
  \]
  then, we obtain
  \[
    2 x\T \left[ Q_{\rho} \otimes PBX\right] w
    \leq \gamma x\T \left[ L\T L \otimes PBB\T P \right] x
    + (\check{v}^N)\T \check{v}^N
  \]
  with
  \begin{equation}\label{checkv}
    \check{v}^N = \sqrt{\nu} \left[  Q_\rho \otimes X\right]w.
  \end{equation}
  Note that since $w$ is bounded, there exists some fixed $\alpha$ such
  that
  \begin{equation}\label{lastone4xxN}
    \sup_{t\in[T,\infty)} \| \check{v}^N(t) \| \leq \alpha.
  \end{equation}
  We get
  \begin{equation}\label{lastone2xxN}
    \dot{V}_N \leq - V_N
    - \gamma x\T \left[ L\T L \otimes PBB\T P \right] x +  (\check{v}^N)\T\check{v}^N,
  \end{equation}
  for $t>T$. Moreover,
  \[
    x\T \left[ L\T L \otimes PBB\T
      P \right] x  \geq  \sum_{i=1}^N \dot{\rho}_i.
  \]
  The above implies
  \begin{equation}\label{ttg1N}
    \dot{V}_N \leq -V_N -\gamma
    \sum_{i=1}^N \dot{\rho}_{i} + 
    (\check{v}^N)\T\check{v}^N. 
  \end{equation} 
  Note that the bound in \eqref{lastone4xxN} combined with the
  inequality \eqref{ttg1N} for $t>T$ implies that there exists some
  $\eta>0$, and $\mu$ such that
  \begin{equation}\label{ttg4N}
    \tilde{\rho}_i(t_2)-\tilde{\rho}_i(t_1) <\eta(t_2-t_1)+\mu,
  \end{equation}
  for $i=1,\ldots, N$ and $t_2,t_1>T$. 

  If \eqref{ttg4N} is satisfied for $t_1,t_2>T$ then clearly the lemma immediately
  follows for $t_1,t_2>T$. Since for $t<T$ all signals are bounded it
  is clear that for $t<T$ the $\rho_i$ can grow at most linearly and
  hence we can obtain the result for all $t_1,t_2>0$.
\end{proof}

\begin{lemma}\label{lem2}
  Consider MAS \eqref{agent-g-noise} with associated network
  communication \eqref{zetanoise}, and the protocol \eqref{protocol}.
  Assume Assumption \ref{ass} is satisfied. Additionally, assume that
  either the $\rho_i$ associated with agents belonging to the basic
  bi-components are bounded or the graph is strongly connected. In that
  case, all $\rho_i$ remain bounded.
\end{lemma}

\begin{proof}[Proof of Lemma \ref{lem2}]
  We prove this result by contradiction. Let $k$ be such that
  $\rho_i$ is unbounded for $i\leq k$ while $\rho_i$ is bounded
  for $i>k$. Clearly, if all $\rho_i$ are unbounded then we have
  $k=N$.
  
  For each $s\in \N$, we define a time-dependent permutation $p$ of
  $\{ 1,\ldots,N\}$ such that
  \begin{equation}\label{perm}
    \rho_{p_s(1)}(s) \geq \rho_{p_s(2)}(s) \geq \rho_{p_s(3)}(s) \geq
    \cdots \geq \rho_{p_s(k)}(s),\quad p_s(i)=i \text{ for } i>k
  \end{equation}
  and we choose $p_t=p_s$ for $t\in[s,s+1)$. Note that Lemma
  \ref{lem2a} implies there exists $T>0$ such that $\tilde{\rho}_k(t)\leq 2 \tilde{\rho}_i(t)$
  for $t>T$ (recall that by construction $\tilde{\rho}_i(t)\rightarrow
  \infty$ for $t\rightarrow \infty$). We define
  \[
    \tilde{x}^i(t) = \hat{x}_{p_t(i)}(t),\qquad
    \tilde{\rho}_i(t) = \hat{\rho}_{p_t(i)}(t),\qquad
    \tilde{\alpha}_i(t) = \alpha_{p_t(i)}(t),\qquad
    \tilde{\zeta}_i(t) = \zeta_{p_t(i)}(t).
  \]
  For $k<N$, we set
  \begin{equation}\label{tildeVk}
    \tilde{V}_k= (\tilde{x}^k)\T (\tilde{\rho}^{-k} \tilde{H}^k \otimes
    P) \tilde{x}^k, 
  \end{equation}
  while for $k=N$ we have
  \begin{equation}\label{tildeVN}
    \tilde{V}_N= (\tilde{x}^N)\T \left[ \tilde{Q}_{\rho} \otimes P 
    \right] \tilde{x}^N,
  \end{equation}
  where
  \begin{equation}\label{tildeQrho}
    \tilde{Q}_{\rho} = \tilde{\rho}^{-N} \left(
      \tilde{H}^N\tilde{\rho}^N - \tilde{\mu}_N
      \tilde{\textbf{h}}_N\tilde{\textbf{h}}_N\T \right)
    \tilde{\rho}^{-N}. 
  \end{equation}  
  Here $\tilde{H}^k, \tilde{\textbf{h}}_N$ and $\tilde{\rho}^N$ are
  obtained from $H^k, \textbf{h}_N$ and $\rho^N$ using our
  permutation.
  
  If we assume $k<N$ then using the arguments of Lemma
  \ref{lem2a}, we obtain \eqref{barxt2xxxy} and the bound
  \eqref{K3K4}. We also obtained in Lemma
  \ref{lem2a} the bound \eqref{lastone2xx}.
  Using the permutation we
  introduced this immediately yields:
  \begin{equation}\label{lastone2}
    \dot{\tilde{V}}_k \leq - \tilde{V}_k 
    - \gamma (\tilde{x}^k)\T \left[ \tilde{L}_{11}\T \tilde{L}_{11}
      \otimes PBB\T P \right] 
    \tilde{x}^k    
    + (\tilde{\check{s}}^k)\T \tilde{\check{s}}^k +
    (\tilde{\check{v}}^k)\T\tilde{\check{v}}^k, 
  \end{equation}
  where $\tilde{\check{s}}^k, \tilde{\check{v}}^k,\tilde{L}_{11}$ are
  obtained from $\check{s}^k, \check{v}^k$ and $L_{11}$ by applying
  the  permutation introduced above.
  
  Note that there exists some fixed $\tilde{\alpha}$ such that
  \begin{equation}\label{lastone4}
    \| \tilde{\rho}_{k} \tilde{\check{s}} \|_2  \leq
    \tilde{\alpha},\qquad 
    \| \tilde{\rho}_{k} \tilde{\check{v}} \|_\infty \leq \tilde{\alpha},
  \end{equation}
  where we exploited that $\tilde{\rho}_k \leq 2\tilde{\rho}_i$ for $i=1,\ldots,
  k$ implies
  \[
    \tilde{\rho}_k \tilde{H}^k \tilde{\rho}^{-k}
  \]
  is bounded.

  On the other hand, for $k=N$, then using the arguments of Lemma \ref{lem2a}, we obtain
  \eqref{lastone2xxN}. Set $\check{w}^N=0$ and let $\check{v}^N$ be
  defined by \eqref{checkv}.  Using the permutation we introduced this
  immediately yields \eqref{lastone2} where
  $\tilde{\check{s}}, \tilde{\check{v}},\tilde{L}_{11}$ are obtained
  from $\check{s}, \check{v}$ and $L_{11}=L$ by applying the
  permutation introduced above.
  
  Next, we consider $k=N$. Note that $\tilde{\rho}_N \leq 2\tilde{\rho}_i$ for $i=1,\ldots
  N$ implies that
  \[
    \tilde{\rho}_N \tilde{Q}_\rho
  \]
  is bounded which yields that there exists some fixed
  $\tilde{\alpha}$ such that \eqref{lastone4} is satisfied
  
  We have established \eqref{lastone2} and\eqref{lastone4} for both
  $k<N$ and $k=N$. Equation \eqref{lastone2} implies
  \begin{equation}\label{lastone}
    \dot{\tilde{V}}_k \leq -\tilde{V}_k + \tilde{\check{s}}\T
    \tilde{\check{s}} + \tilde{\check{v}}\T\tilde{\check{v}}. 
  \end{equation}
  This clearly yields that $\tilde{V}_k$ is bounded given our bounds
  \eqref{lastone4}. Using inequality \eqref{ttg1} together with our
  permutation we obtain that, for $t>T_1$, we have
  \begin{align}
    \left[ \tilde{\rho}_{k}^2 \tilde{V}_k \right]' &\leq
    2\tilde{\rho}_{k} \dot{\tilde{\rho}}_{k} \tilde{V}_k -
    \tilde{\rho}_{k}^2 \tilde{V}_k -\gamma 
    \tilde{\rho}_{k}^2 \dot{\tilde{\rho}}_{k} + \tilde{\rho}_{k}^2 (\tilde{\check{s}}\T
    \tilde{\check{s}} + \tilde{\check{v}}\T\tilde{\check{v}}) \nonumber\\
    &\leq  - \tilde{\rho}_{k}^2 \tilde{V}_k + \tilde{\rho}_{k}^2 (\tilde{\check{s}}\T
    \tilde{\check{s}} + \tilde{\check{v}}\T\tilde{\check{v}}),  \label{lastone3}
  \end{align}
  where we choose $T_1>T$ such that
  $2\tilde{V}_k \leq \gamma \tilde{\rho}_{k}$ for $t>T_1$ which is
  obviously possible since $\tilde{\rho}_{k}$ increases to infinity
  while $\tilde{V}_k$, as argued before, is bounded. Then, using
  \eqref{lastone4} and \eqref{lastone3} we find
  \[
    \tilde{\rho}_{k}^2(s+\sigma) \tilde{V}_k(s+\sigma) < e^{-\sigma}
    \tilde{\rho}_{k}^2(s) \tilde{V}_k(s)  + 2\tilde{\alpha}^2, 
  \]
  for all $\sigma\in (0,1]$ and any $s\in \N$ with $s>T_1$. This by
  itself does not yield that $\tilde{\rho}_{k}^2 \tilde{V}_k$ is
  bounded because we have potential discontinuities for $s\in \N$
  because of the reordering process we introduced. Note that
  $\tilde{V}_k$ is not affected by the reordering but $\tilde{\rho}_k$
  can have discontinuities  for $s\in \N$. Hence
  \[
    \tilde{\rho}_{k}^2(s^+)  \text{ and }
    \tilde{\rho}_{k}^2(s^-) 
  \]
  might be different and, strictly speaking, we have obtained
  \begin{equation}\label{yhju}
    \tilde{\rho}_{k}^2(s+\sigma) \tilde{V}_k(s+\sigma) < e^{-\sigma}
    \tilde{\rho}_{k}^2(s^+) \tilde{V}_k(s)  + 2\tilde{\alpha}^2,
  \end{equation}  
  for all $\sigma\in (0,1]$ and any $s\in \N$ with $s>T_1$. 

  Given the bounds on the growth of $\rho_i$ obtained in Lemma
  \ref{lem2a}, it is easy to see that there exists $A_0>0$ such that a
  discontinuity can only occur when
  \[
    \tilde{\rho}_{k-1}(s)-\tilde{\rho}_{k}(s)<A_0,
  \]
  for $s$ sufficiently large. Using this bound, together with the fact
  that $\tilde{\rho}_k$ converges to infinity, we find that for any
  $\eps$ there exists $T_2>T_1$ such that
  \begin{equation}\label{bound2}
    \tilde{\rho}_k^2(s^+)  \leq (1+\eps)
    \tilde{\rho}_k^2(s^-),
  \end{equation}
  for $s>T_2$.  Therefore, we find that
  \begin{equation}\label{ttg3}
    \tilde{\rho}_{k}^2(t) \tilde{V}_k(t) 
  \end{equation}
  is bounded by combining \eqref{bound2} with \eqref{yhju} provided
  that we choose $(1+\eps)e^{-1} < 1$.
  
  In addition, \eqref{ttg1} implies
  \begin{equation}\label{ttg2}
    \dot{\tilde{V}}_k \leq -\tilde{V}_k-\gamma
    (\tilde{x}^k)\T \left[ \tilde{L}_{11}^2 \otimes PBB\T
      P \right] \tilde{x}^k + \tilde{\check{s}}\T \tilde{\check{s}} +
    \tilde{\check{v}}\T\tilde{\check{v}},
  \end{equation}
  which, combined with  \eqref{lastone4},
  implies that
  \begin{equation*}
    \gamma \int_{s}^{s+1} \tilde{\rho}_{k}^2 (\tilde{x}^k)\T \left[
      \tilde{L}_{11}^2 \otimes PBB\T P \right] \tilde{x}^k\,
    \textrm{d}t \leq\\
    4 \tilde{\rho}_{k}^2(s) \tilde{V}_k(s) + 4
    \int_{s}^{s+1} \tilde{\rho}_{k}^2
    (\tilde{\check{s}}\T\tilde{\check{s}}+\tilde{\check{v}}\T\tilde{\check{v}})
    \textrm{d}t\leq 4\tilde{\rho}_{k}^2(s) \tilde{V}_k(s) + 8\tilde{\alpha}^2
  \end{equation*}
  for all $s$ since $\tilde{\rho}_k(t)\leq 2\tilde{\rho}_k(s)$ for
  $t\in [s,s+1]$. Boundedness of \eqref{ttg3} then implies that
  \[
    \int_{s}^{s+1} \tilde{\rho}_{k}^2 (\tilde{\zeta}^k)\T (I\otimes PBB\T
    P) \tilde{\zeta}^k\, \textrm{d}t
  \]
  is bounded. 

  We have established that $\tilde{\rho}_k^2 \tilde{V}_k$ is bounded
  but this does not establish that the $\tilde{x}^k$ gets small since
  the $\tilde{\rho}_i$ for $i=1,\ldots,k-1$ might be much larger than
  $\tilde{\rho}_k$. We need to do some extra work.
  We have that \eqref{barxt2xxxy} implies
  \begin{equation}\label{barxt2xxxy2}
    \dot{\tilde{x}}^k = (I\otimes A)\tilde{x}^k
    -[\tilde{\rho}^{k}
    \tilde{L}_{11} \otimes BB\T P]\tilde{x}^k
    +[I \otimes B](\tilde{s}^k+\tilde{v}^k),
  \end{equation}
  where $\tilde{s}^k,\tilde{v}^k,\tilde{L}_{11}$ are
  obtained from $\hat{s}^k,\hat{v}^k,\bar{L}_{11}$ by
  applying the permutation introduced above. A permutation clearly
  does not affect the bound we obtained in 
  \eqref{K3K4} and we obtain
  \begin{equation}\label{K3K42}
    \| \tilde{s}^k \|_\infty < K_3,\qquad \| \tilde{v}^k \|_2 < K_4.
  \end{equation}
  For any $j<k$, we decompose
  \begin{equation}\label{tildeL11}
    \tilde{x}^j_I = \begin{pmatrix}
      \tilde{x}_1 \\ \vdots \\ \tilde{x}_j \end{pmatrix},\quad
    \tilde{x}^j_{II} = \begin{pmatrix}
      \tilde{x}_{j+1} \\ \vdots \\ \tilde{x}_k \end{pmatrix},\quad
    \tilde{L}_{11} =\begin{pmatrix} \tilde{L}^j_{11} & \tilde{L}^j_{12} \\
      \tilde{L}^j_{21} & \tilde{L}^j_{22}
    \end{pmatrix},
  \end{equation}
  with $\tilde{L}^j_{11}\in \R^{j\times j}$ while
  \begin{equation*}
    \tilde{s}^k = \begin{pmatrix} \tilde{s}^j \\
      \tilde{s}^j_c \end{pmatrix},\quad
    \tilde{v}^k = \begin{pmatrix} \tilde{v}^j \\
      \tilde{v}^j_c \end{pmatrix},
  \end{equation*}
  with 
  $\tilde{s}^j\in \R^{nj}$, $\tilde{v}^j\in \R^{nj}$
  and
  \begin{equation}\label{checkx}
    \check{x}^j = \tilde{x}^j_I  + (\tilde{L}_{11}^j)^{-1}
    \tilde{L}^j_{12} \tilde{x}^j_{II},
  \end{equation}
  for $j<k$ while $\check{x}^k=\tilde{x}^k$.  We will show that 
  \begin{equation}\label{upbound}
    \tilde{\rho}_j^2 \tilde{V}_j
  \end{equation}
  is bounded for $j=1,\ldots, k$  where
  \begin{equation}\label{tildeVj}
    \tilde{V}_j=(\check{x}^j)\T \left[ \tilde{H}^j \tilde{\rho}^{-j} \otimes P
  \right] \check{x}^j, 
  \end{equation}
  while
  \[
    \tilde{\rho}^j=\begin{pmatrix}
      \tilde{\rho}_1 & 0      & \cdots & 0 \\
      0      & \tilde{\rho}_2 & \ddots & \vdots \\
      \vdots & \ddots & \ddots & 0 \\
      0      & \cdots & 0      & \tilde{\rho}_j
    \end{pmatrix},\quad
    \tilde{\rho}_c^j=\begin{pmatrix}
      \tilde{\rho}_{j+1} & 0      & \cdots & 0 \\
      0      & \tilde{\rho}_2 & \ddots & \vdots \\
      \vdots & \ddots & \ddots & 0 \\
      0      & \cdots & 0      & \tilde{\rho}_k
    \end{pmatrix},\quad
    \tilde{H}^j=\begin{pmatrix}
      \tilde{\alpha}_1 & 0      & \cdots & 0 \\
      0      & \tilde{\alpha}_2 & \ddots & \vdots \\
      \vdots & \ddots & \ddots & 0 \\
      0      & \cdots & 0      & \tilde{\alpha}_j
    \end{pmatrix}.
  \]
  Note that if $k=N$ then we have \eqref{tildeVN} instead of \eqref{tildeVj}.
  We also define
  \[
    \tilde{\zeta}^j=\begin{pmatrix} \tilde{\zeta}_1 \\ \vdots \\
      \tilde{\zeta}_j \end{pmatrix},\qquad
    \tilde{\zeta}^j_c=\begin{pmatrix} \tilde{\zeta}_{j+1} \\ \vdots \\
      \tilde{\zeta}_k \end{pmatrix}.
  \]
  It is not hard to verify that if $\tilde{\rho}_j^2\tilde{V}_j$ is bounded for
  $j=1,\ldots, k$ then we have that  $\hat{x}_j$ is arbitrary small for
  large $t$ since the $\tilde{\rho}_j$ are increasing and converge to
  infinity. On the other hand, $\hat{x}_j$ arbitrary small for large
  $t$ and for all $j=1,\ldots, k$ would imply that the $\rho_j$ are
  all constant for large $t$ which contradicts our premise that the
  $\rho_j$ are unbounded for $j=1,\ldots, k$.
  
  We will establish boundedness of $\tilde{\rho}_j^2\tilde{V}_j$ via
  recursion. Assume for $i=j$ we have either
  \begin{equation}\label{ttg3b}
    \tilde{\rho}_i^2 \tilde{V}_i
  \end{equation}
  is unbounded or
  \begin{equation}\label{ttg3c}
    \int_{s}^{s+1} \tilde{\rho}_i^2 (\tilde{\zeta}^i)\T (I
    \otimes PBB\T P) \tilde{\zeta}^i\, \textrm{d}t
  \end{equation}
  is unbounded while, for $j<i\leq k$, both \eqref{ttg3b} and
  \eqref{ttg3c} are bounded. We will show this yields a
  contradiction. Note that in the above we already established
  \eqref{ttg3b} and \eqref{ttg3c} are bounded for $i=k$.
  
  Using \eqref{barxt2xxxy2} and \eqref{checkx}, we obtain
  \begin{equation}\label{checkxj}
    \dot{\check{x}}^j = (I\otimes A) \check{x}^j -
    \left[\tilde{\rho}^j\tilde{L}_{11}^j\otimes
      BB\T P \right] \check{x}^j+(I\otimes B)(\check{s}+\check{v})
    + (W_j \tilde{\rho}^j_c \otimes BB\T P)\tilde{\zeta}_c^j,
  \end{equation}
  where
  \begin{align*}
    \check{s}^j & =\tilde{s}^j+(\tilde{L}_{11}^j)^{-1}\tilde{L}_{12}^j
    \tilde{s}_j^c,\\
    \check{v}^j & =\tilde{v}^j+(\tilde{L}_{11}^j)^{-1}\tilde{L}_{12}^j
    \tilde{v}^j_c + \left[ (\tilde{L}^j_{11})^{-1}\tilde{L}^j_{12}
    \tilde{\rho}^j_c \otimes BB\T P \right] \tilde{\zeta}_c^j. 
  \end{align*}
  Finally,
  \eqref{ttg3c} implies $\tilde{\rho}_i B\T P\tilde{\zeta}_i$ has
  bounded energy for $i=j+1,\ldots ,k$ which in turn yields
  $(\tilde{\rho}_c^j\otimes B\T P)\tilde{\zeta}_c^j$ has bounded
  energy. Clearly, $\check{s}^j$ is bounded while $\check{v}^j$ has
  bounded energy by \eqref{K3K42}
  and that there exist constants such that
  \begin{equation}\label{tildeK3K4}
    \| \check{s}^j \|_\infty < \tilde{K}_3,\quad
    \int_s^{s+1} (\check{v}^j)\T\check{v}^j \textrm{d}t < \tilde{K}_3.
  \end{equation}
  We obtain
  \begin{equation*}
    \dot{\tilde{V}}_j =(\check{x}^j)\T \left[
      \tilde{H}^j\tilde{\rho}^{-j} \otimes (-I +PBB\T
      P) \right] \check{x}^j \\
    -(\check{x}^j)\T \left[ (\tilde{H}^j\tilde{L}_{11}^j+\tilde{L}_{11}^j
      \tilde{H}^j)\otimes PBB\T P\right]\check{x}^j \\
    + 2(\check{x}^j)\T \left[ \tilde{H}^j\tilde{\rho}^{-j} \otimes PB \right]
    (\check{s}^j+\check{v}^j),
  \end{equation*}
  using \eqref{tildeVj} and \eqref{checkxj}. \eqref{HkL11} implies
  \[
    \tilde{H}^j\tilde{L}_{11}^j+\tilde{L}_{11}^j
    \tilde{H}^j > 3\gamma (\tilde{L}_{11}^j)\tilde{L}_{11}^j.
  \]
  Moreover, we have
  \[
    \tilde{\rho}_j \tilde{H}^j \tilde{\rho}^{-j}  \leq \gamma,
  \]
  for some $\gamma>0$ using \eqref{perm}. This yields that there
  exists $\eta>0$ such that
  \begin{equation}\label{gettingthere}
    \dot{\tilde{V}}_j \leq - \tilde{V}_j -\gamma
    (\check{x}^j)\T \left[ (\tilde{L}_{11}^j)\T \tilde{L}_{11}^j \otimes PBB\T
      P\right]\check{x}^j+\tfrac{\eta}{\tilde{\rho}_j^2}
    [(\check{s}^j)\T 
    \check{s}^j+(\check{v}^j)\T \check{v}^j].
  \end{equation}
  Similarly to our derivation of \eqref{lastone3} we choose $T_1$ such that
  $2\tilde{V}_j \leq \gamma \tilde{\rho}_j$. This yields that,
  \eqref{tildeK3K4} and \eqref{gettingthere} imply
  \begin{equation}\label{llast}
    \tilde{\rho}_{j}^2(s+\sigma) \tilde{V}_j(s+\sigma) < e^{-\sigma}
    \tilde{\rho}_{j}^2(s^+) \tilde{V}_j(s^{+})  + 2\tilde{K}_3^2\eta,
  \end{equation}
  for all $\sigma\in (0,1]$, and any $s\in \N$ wih $s>T_1$. This by
  itself does not imply that $\tilde{\rho}_{j}^2(s) \tilde{V}_j(s)$ is
  bounded because at time $s\in \N$ there might be a discontinuity due
  to the reordering we performed.

  If a new permutation has the same $j$ agents with the largest
  $\rho_i$ then $\tilde{V}_j(s^+) = \tilde{V}_j(s^-)$ and we obtain,
  as before, that there exists some $\eps$ satisfying
  $(1+\eps)e^{-1}<1$ and a $T_2>T_1$ such that
  \begin{equation}\label{llast23}
    \tilde{\rho}_j^2(s^+) \tilde{V}_j(s^+) \leq (1+\eps)
    \tilde{\rho}_j^2(s^-) \tilde{V}_j(s^-),
  \end{equation}
  for $s>T_2$ similarly as we did in the derivation of equation \eqref{bound2}.

  If a new permutation changes the set of $j$
  agents with the largest $\rho_i$ then it is easy to see that
  there exists $A_0>0$ such that a 
  discontinuity can only occur when
  \[
    \tilde{\rho}_j(s)-\tilde{\rho}_{j+1}(s)<A_0
  \]
  for $s$ sufficiently large using Lemma \ref{lem2a}. This implies
  that there exists some $A>0$ such that
  \begin{equation*}
    \tilde{\rho}_{j}^2(s)\tilde{V}_j(s) < (\tilde{\rho}_{j+1}(s)+A_0)^2 \tilde{V}_j(s)\\
    < (\tilde{\rho}_{j+1}(s)+A_0)^2 M \tilde{V}_{j+1}(s)\\
    < \frac{(\tilde{\rho}_{j+1}(s)+A_0)^2}{\tilde{\rho}_{j+1}^2(s)}
    M \tilde{\rho}_{j+1}^2(s) \tilde{V}_{j+1}(s) < A,
  \end{equation*}
  for large $s$ using Lemma \ref{2.10} since we already established that
  $\tilde{\rho}_{j+1}^2 \tilde{V}_{j+1}(s)$ is bounded while
  $\tilde{\rho}_{j+1}$ is increasing to infinity. In that case
  \eqref{llast} shows
  \[
    \tilde{\rho}_{j}^2(s+\sigma) \tilde{V}_j(s+\sigma)
  \]
  is bounded as well. If $\tilde{\rho}_j^2(s) \tilde{V}_j(s)$ is
  larger than $A$ then we know discontinuities of $\tilde{V}_j(s)$ do
  not arise and hence \eqref{llast} and \eqref{llast23} show that
  $\tilde{\rho}_j^2 \tilde{V}_j$ remains bounded. It remains to show that
  \eqref{ttg3c} is bounded. This follows immediately from
  \eqref{gettingthere} in combination with \eqref{tildeK3K4} and the
  boundedness of $\tilde{\rho}_j^2 \tilde{V}_j$.
	
  In this way, we recursively established that \eqref{ttg3b} is
  bounded for $i=1,\ldots,k$ and, as noted before, this implies that
  $\rho_i(t)$ is constant for large $t$ and $i=1,\ldots, k$ which
  contradicts our assumption that some of the $\rho_i$ are
  unbounded. This completes our proof.
\end{proof}

\section{\textbf{Numerical examples}}

In this section, we will show that the proposed protocol achieves
scalable $\delta-$level coherent state synchronization. We study the
effectiveness of our proposed protocol as it is applied to systems
with different sizes, different communication graphs, different noise
patterns, and different $\delta$ values. In all examples of the paper, the weight
of edges of the communication graphs is considered to be equal~$1$.

We consider agent models as
\begin{equation}\label{agent-g-noise-Example}
  \begin{system*}{l}
    \dot{x}_i(t)=\begin{pmatrix}0&1&0\\0&0&1\\0&0&0\end{pmatrix}x_i(t)+\begin{pmatrix}
      0\\0\\1
    \end{pmatrix}u_i(t)+\begin{pmatrix}
      0\\0\\1
    \end{pmatrix}w_i(t),
  \end{system*} 
\end{equation}
for $i=1,\hdots, N.$ We utilize adaptive protocol \eqref{protocol} as following
\begin{equation}\label{protocol-example}
  \begin{system*}{ccl}
    \dot{\rho}_i &=& \begin{cases}
      \zeta_i\T\begin{pmatrix}
        1&2.41&2.41\\
        2.41&5.82&5.82\\
        2.41&5.82&5.82\\
      \end{pmatrix}\zeta_i & \text{ if } \zeta_i\T P \zeta_i \geq d,
      \\
      0 & \text{ if } \zeta_i\T P \zeta_i < d,
    \end{cases} \\
    u_i &=& -\rho_i \begin{pmatrix}
      1&2.41&2.41\\
    \end{pmatrix} \zeta_i,
  \end{system*}
\end{equation}
where $P$ is the solution of the algebraic Riccati equation
\eqref{eq-Riccati} and equals
\[
  P=\begin{pmatrix}
    2.41&2.41&1\\
    2.41&4.82&2.41\\
    1&2.41&2.41\\
  \end{pmatrix}.
\]
\subsection{\textbf{Scalability -- independence to the size of communication networks}}\label{Size}

In this section, we consider MAS with agent models
\eqref{agent-g-noise-Example} and disturbances
\begin{equation}\label{w_i}
  w_i(t)=0.1\sin(0.1it+0.01t^2), \quad i=1,\hdots, N.
\end{equation}
In the following, to illustrate the scalability of proposed protocols,
we study three MAS with $5$, $25$, and $121$ agents communicating over
directed and undirected Vicsek fractal graphs shown in Figure
\ref{vicsek-fractal-graph}. In both following examples, we consider $d=0.5$.

\begin{figure}[ht!]
  \includegraphics[width=0.75\textwidth]{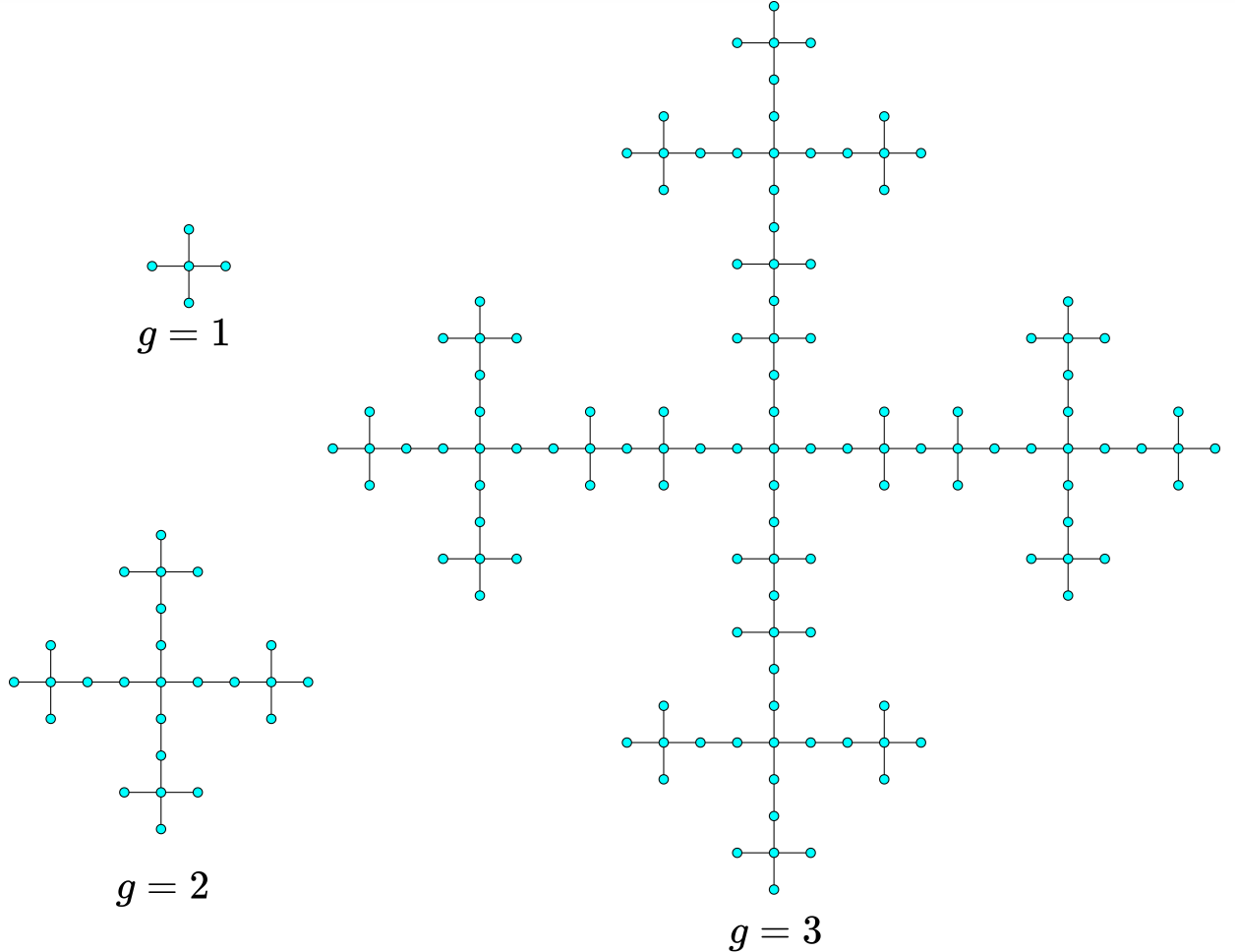}
  \centering
  \caption{Vicsek fractal graphs}\label{vicsek-fractal-graph}
\end{figure}
\begin{figure}[ht!]
  \includegraphics[width=0.1\textwidth]{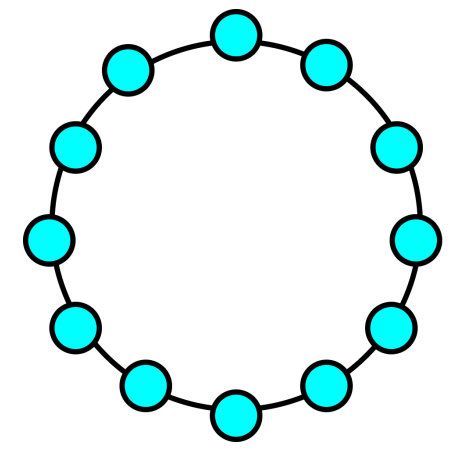} \centering
  \caption{Circulant graph}\label{Circulant_Graph_UD}
\end{figure}

\begin{figure}[ht!]
 \includegraphics[width=0.4\textwidth]{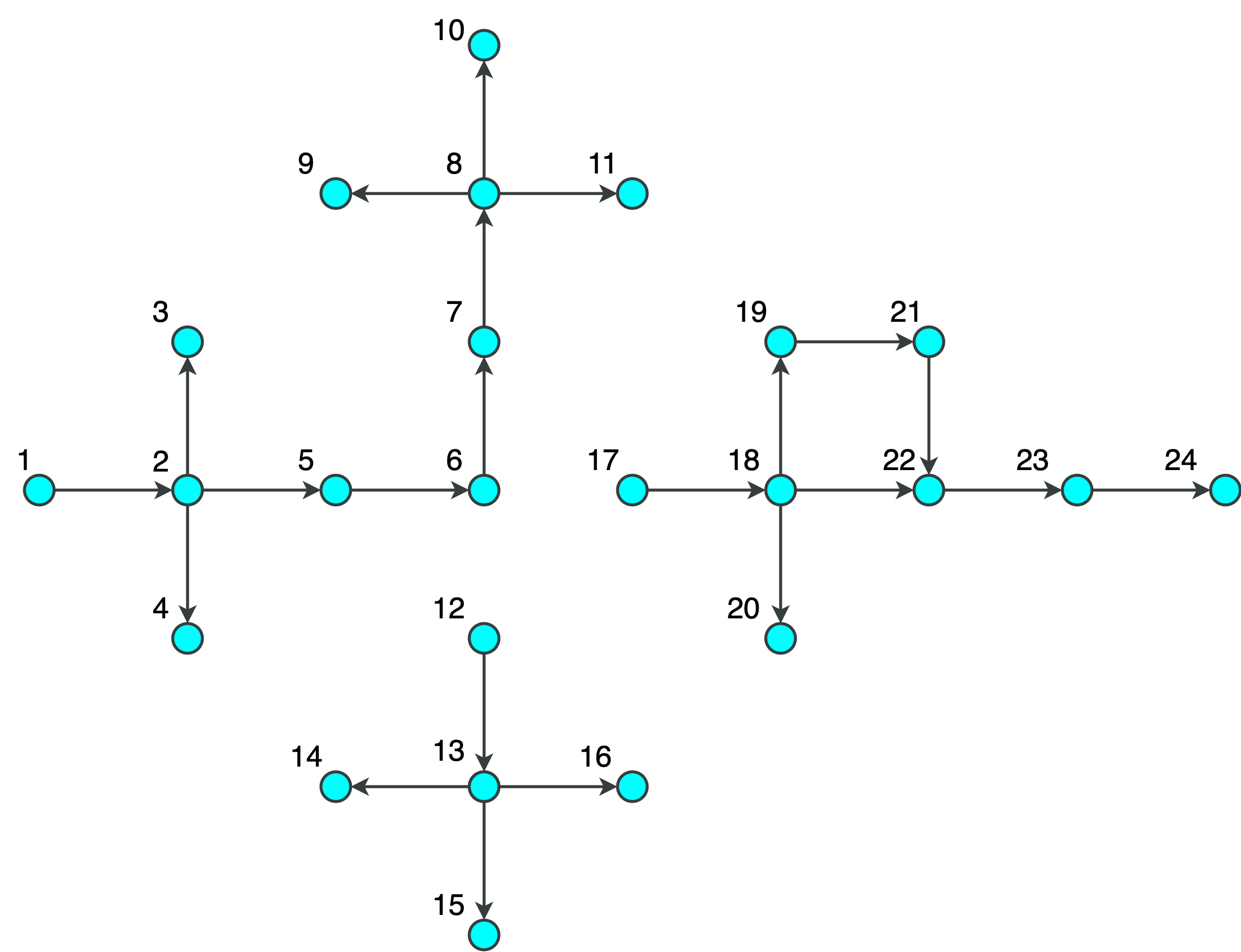} \centering
  \caption{Disconnected directed graph}\label{dis-directed-net}
\end{figure}

\begin{figure}[ht!]
\includegraphics[width=0.4\textwidth]{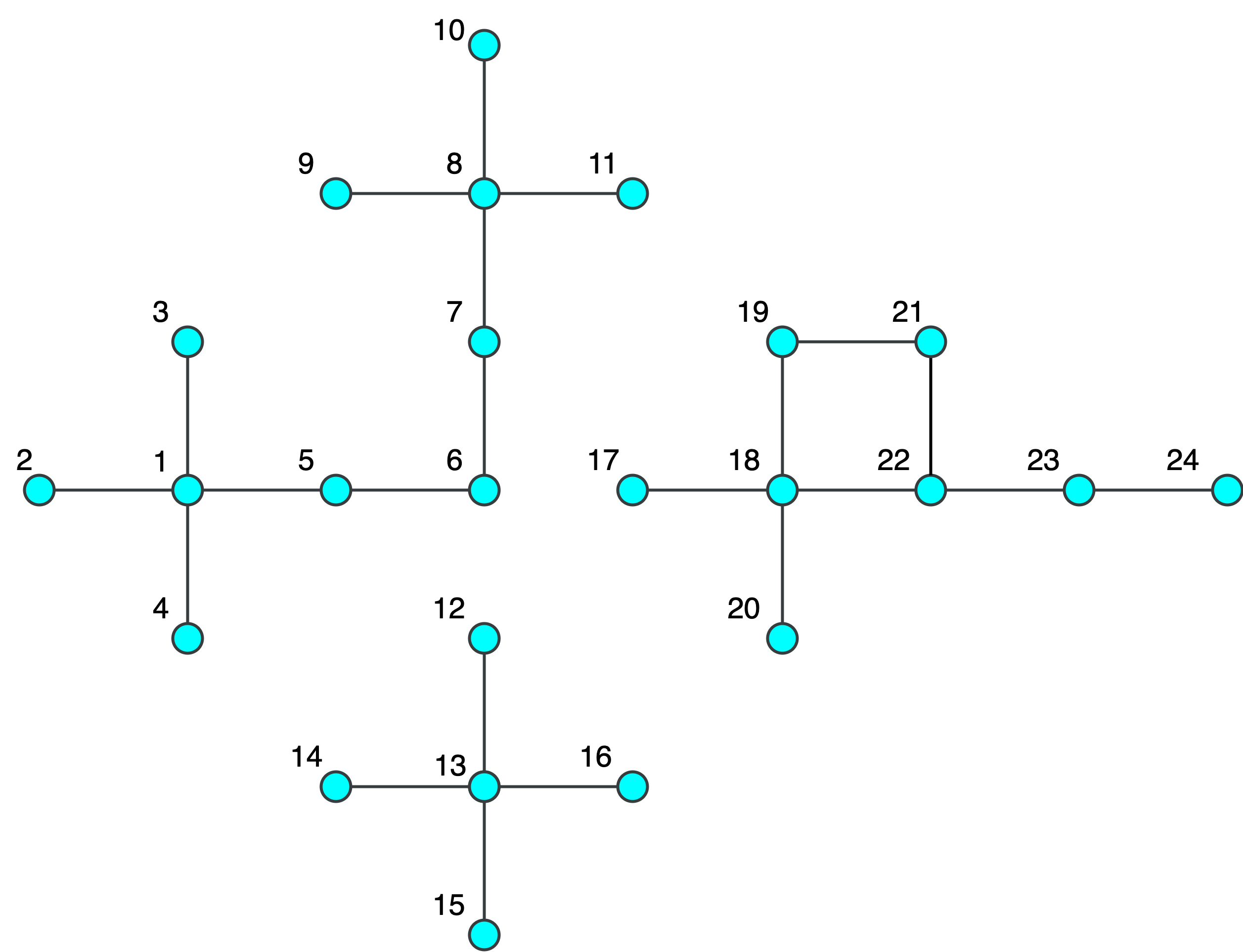} \centering
  \caption{ Disconnected undirected graph}\label{dis-undirected-net}
\end{figure}

\subsubsection{Directed graphs}
First, we consider directed Vicsek fractal graphs. The simulation
results presented in Figures~\ref{figure3a}-\ref{figure3c} clearly
show the scalability of our one-shot-designed protocol. In other
words, the scalable adaptive protocols achieve $\delta-$level coherent
state synchronization independent of the size of the network.

\begin{figure}[p!]
  \centering \includegraphics[width=0.75\textwidth]{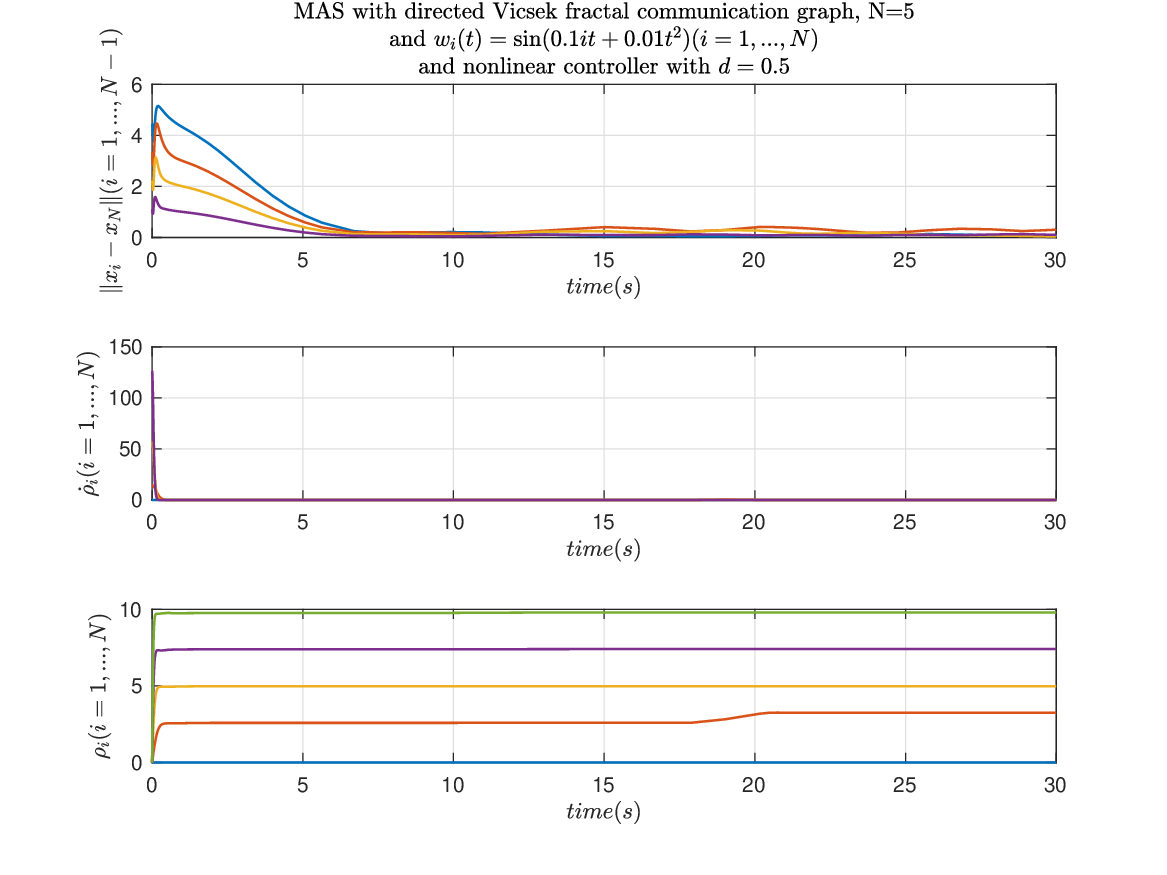}
  \vspace{-3mm}
  \caption[]{{$\delta$-level coherent state synchronization of MAS
      with directed Vicsek fractal communication graphs and $N=5$ in
      the presence of disturbances via protocol
      \eqref{protocol-example} with $d=0.5$}} \label{figure3a}
\end{figure}

\begin{figure}[p!]
	\centering
	\includegraphics[width=0.75\textwidth]{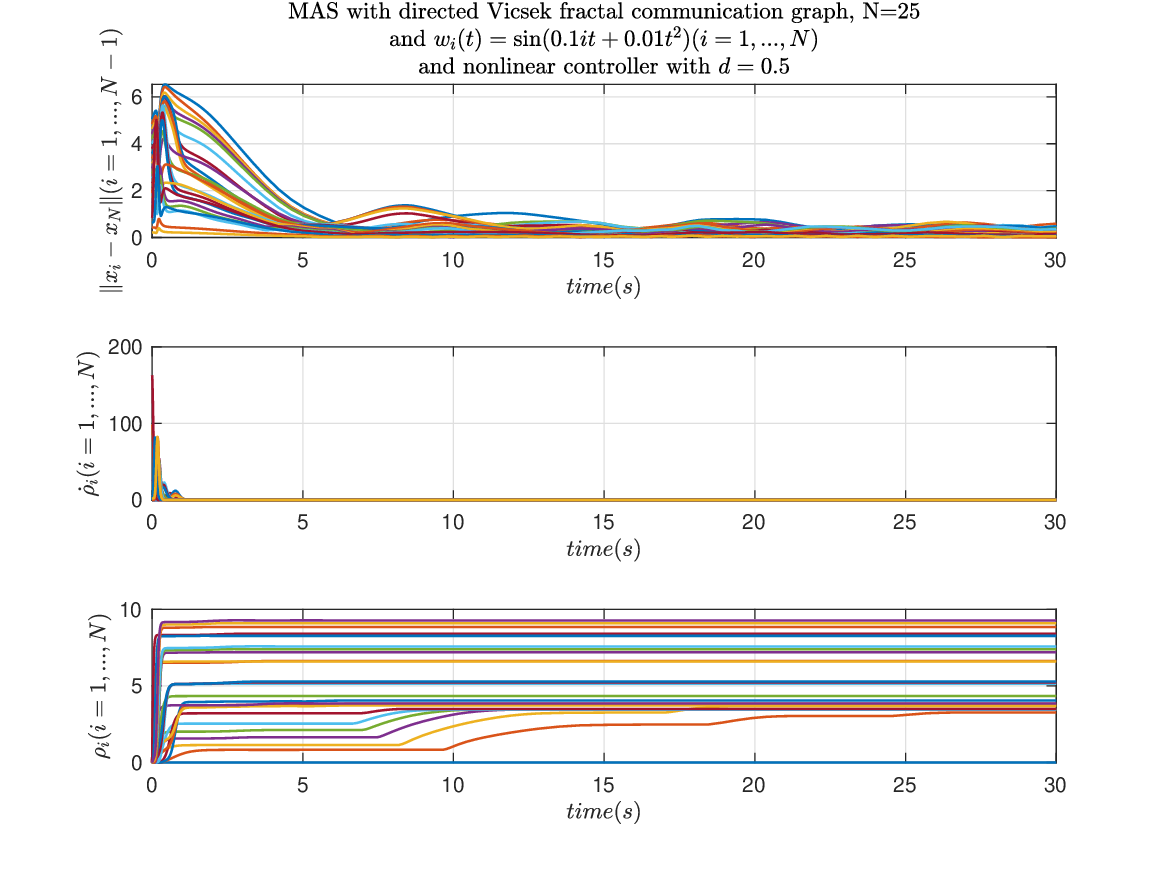}
  \vspace{-3mm}
  \caption[]{{$\delta$-level coherent state synchronization of MAS
      with directed Vicsek fractal communication graphs and $N=25$ in
      the presence of disturbances via protocol
      \eqref{protocol-example} with $d=0.5$}} \label{figure3b}
\end{figure}
\begin{figure}[p!]
  \centering \includegraphics[width=0.75\textwidth]{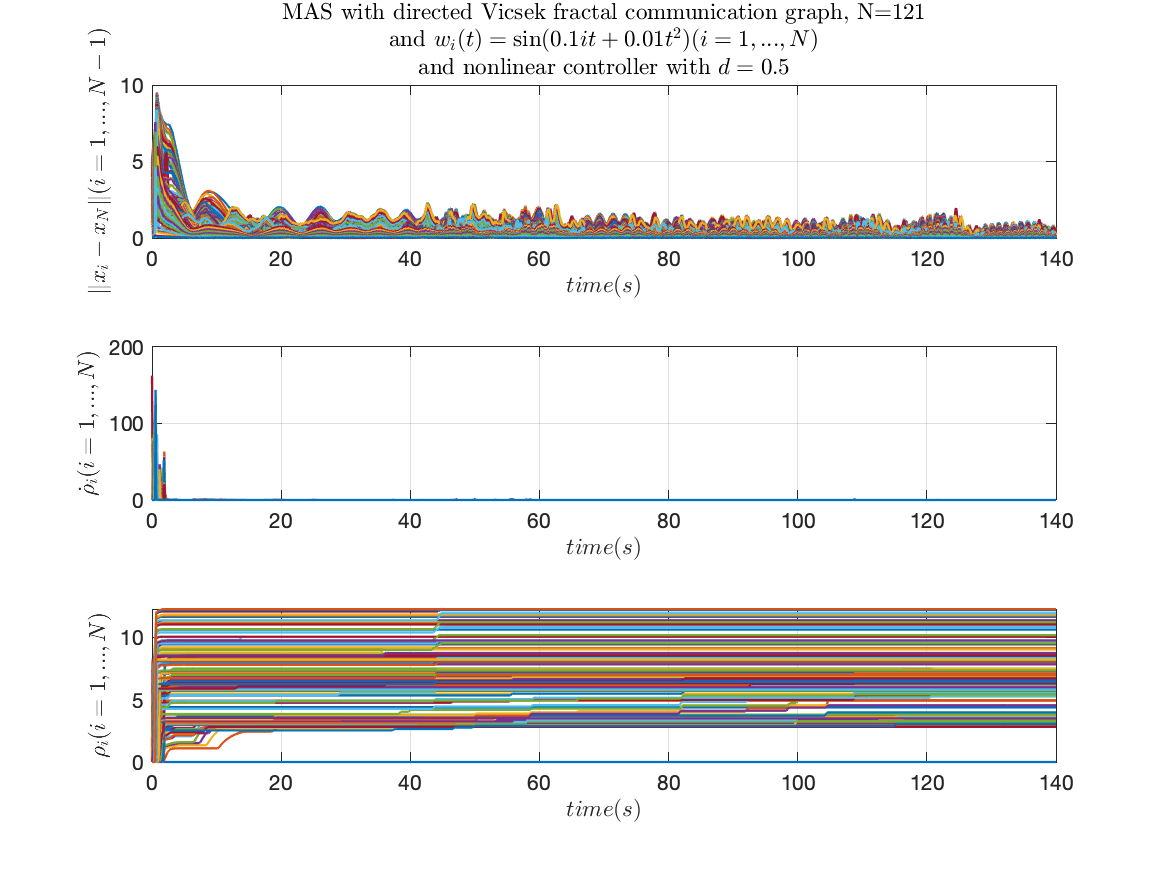}
  \vspace{-3mm}
  \caption[]{{$\delta$-level coherent state synchronization of MAS
      with directed Vicsek fractal communication graphs and $N=121$ in
      the presence of disturbances via protocol
      \eqref{protocol-example} with $d=0.5$}} \label{figure3c}
\end{figure}

\begin{figure}[p!]
  \centering
  \includegraphics[width=0.75\textwidth]{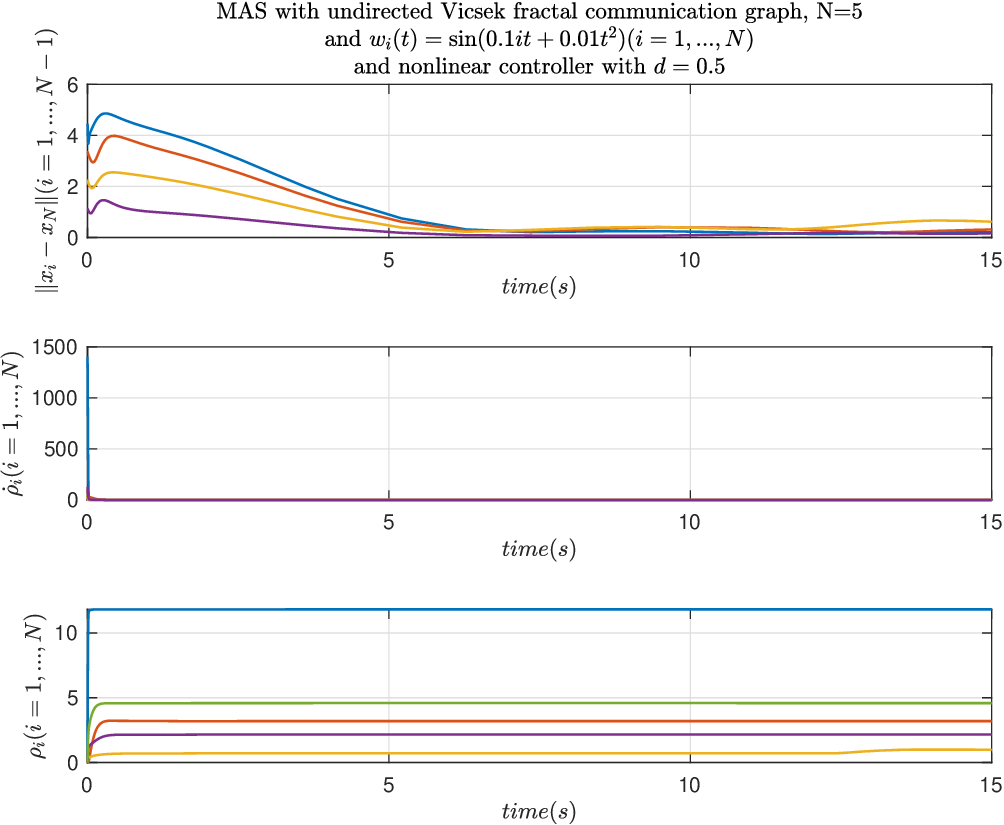}
  \vspace{-3mm}
  \caption[]{{$\delta$-level coherent state synchronization of MAS
      with directed Vicsek fractal communication graphs and $N=5$ in
      the presence of disturbances via protocol
      \eqref{protocol-example} where $d=0.5$}} \label{figure4a}
\end{figure}
\begin{figure}[p!]
  \centering
  \includegraphics[width=0.75\textwidth]{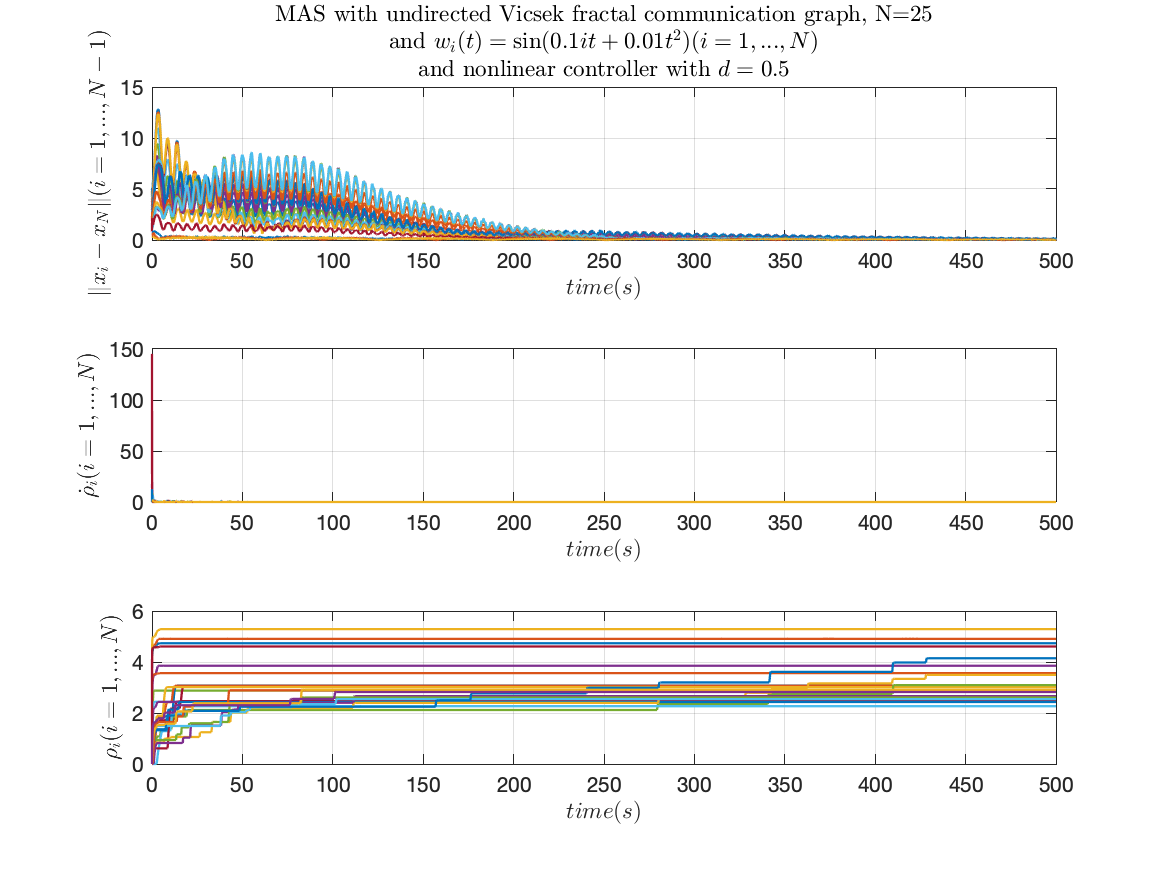}
  \vspace{-3mm}
  \caption[]{{$\delta$-level coherent state synchronization of MAS
      with directed Vicsek fractal communication graphs and $N=25$ in
      the presence of disturbances via protocol
      \eqref{protocol-example} where $d=0.5$}} \label{figure4b}
\end{figure}
\begin{figure}[p!]
  \centering
  \includegraphics[width=0.75\textwidth]{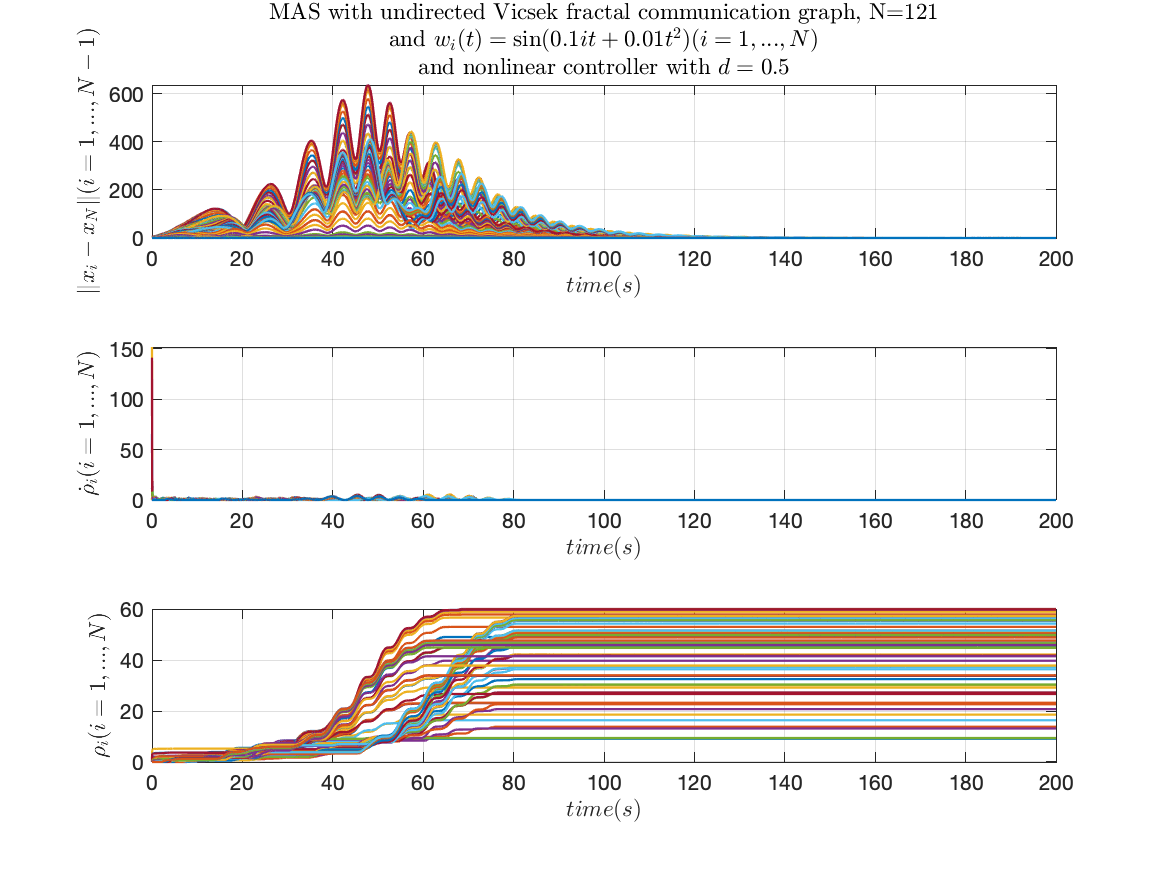}
  \vspace{-3mm}
  \caption[]{{$\delta$-level coherent state synchronization of MAS
      with directed Vicsek fractal communication graphs and $N=121$ in
      the presence of disturbances via protocol
      \eqref{protocol-example} where $d=0.5$}} \label{figure4c}
\end{figure}

\begin{figure}[p!]
  \centering \includegraphics[width=0.75\textwidth]{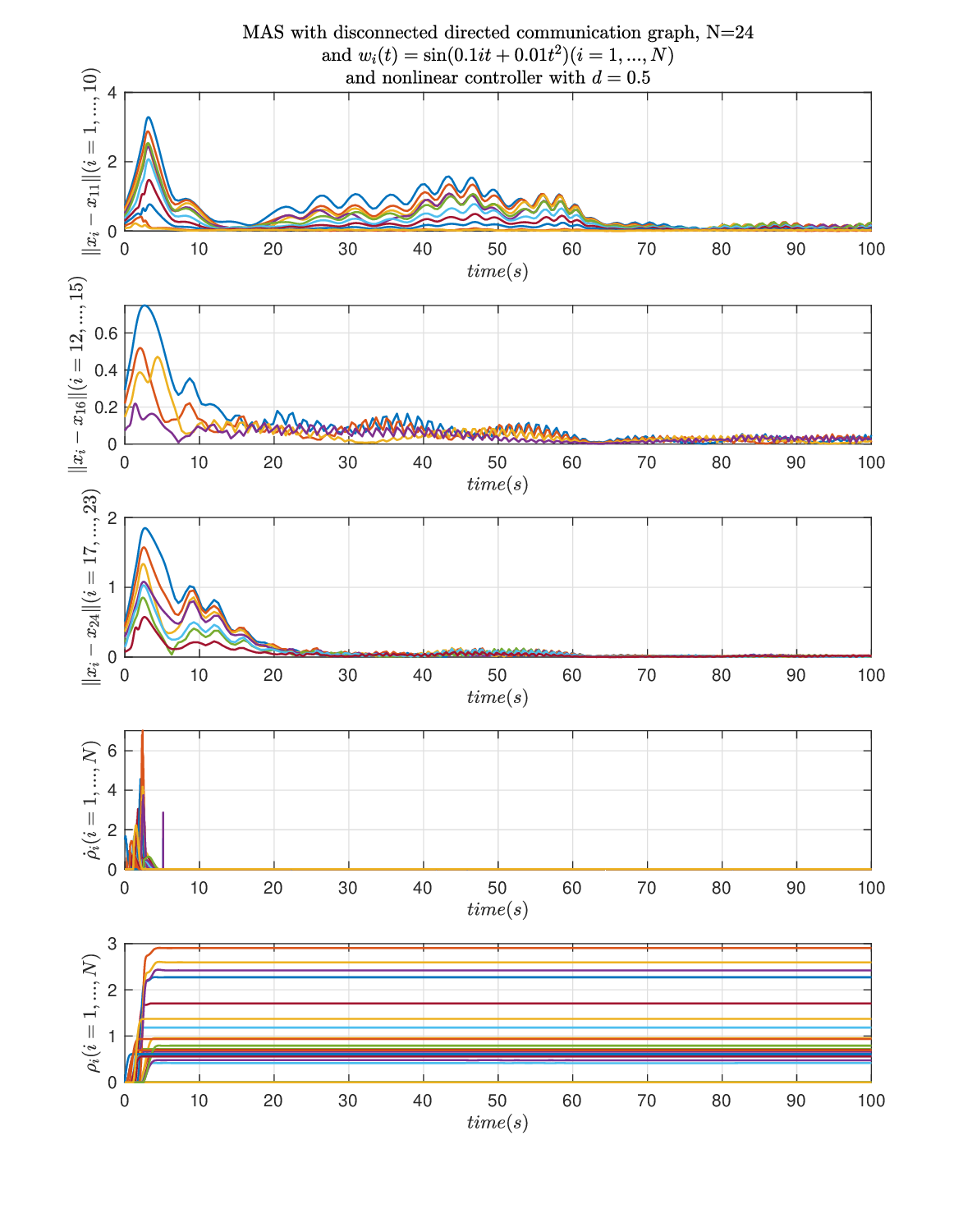}
  \vspace{-3mm}
  \caption[]{{$\delta$-level coherent state synchronization of MAS
      with $N=24$ and disconnected directed communication graphs in the
      presence of disturbances via nonlinear protocol with
      $d=0.5$}} \label{figure10}
\end{figure}

\begin{figure}[p!]
  \centering \includegraphics[width=0.75\textwidth]{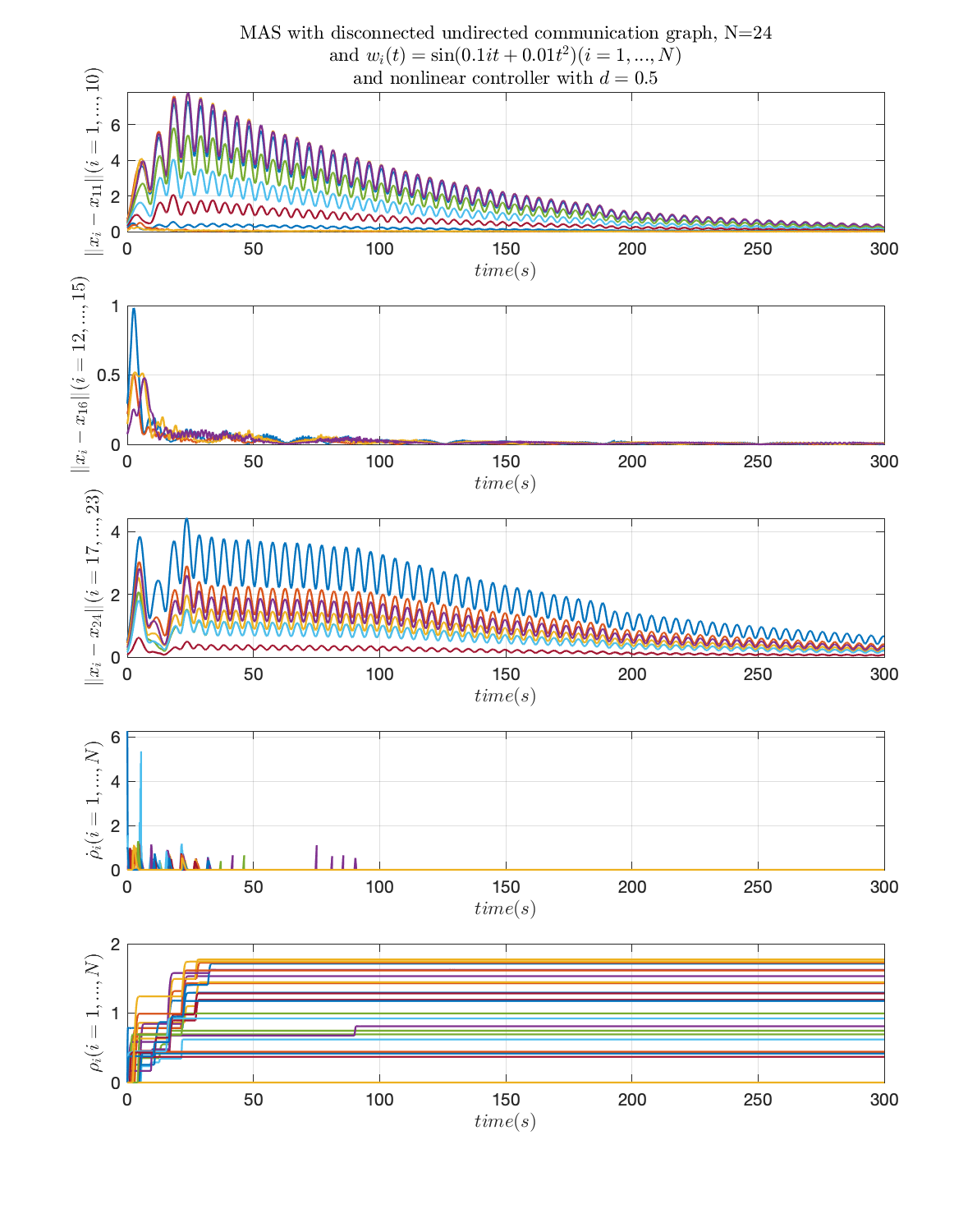}
  \vspace{-3mm}
  \caption[]{{$\delta$-level coherent state synchronization of MAS
      with $N=24$ and disconnected undirected communication graphs in the
      presence of disturbances via nonlinear protocol with
      $d=0.5$}} \label{figure9}
\end{figure}

\subsubsection{Undirected graphs}

Second, we consider undirected Vicsek fractal graphs. The algebraic
connectivity of the undirected Vicsek fractal graphs is presented in
Table \ref{table_vicsek}. It can be easily seen that the size of the
graphs increase the algebraic connectivity of the associated Laplacian
matrix decreases. The simulation results presented in
Figure~\ref{figure4a}-\ref{figure4c} show that the one-shot designed
protocol \eqref{protocol-example}, achieves $\delta-$level coherent
state synchronization regardless of the number of agents and the
algebraic connectivity of the associated Laplacian matrices of the
graphs.
\begin{table}[t!]
  \caption{Algebraic connectivity of undirected Vicsek fractal graphs}
  \centering
  \begin{tabularx}{.3\columnwidth}{X X l}
    \hline
    $N$& g&$\re\{\lambda_2\}$\\
    \hline
    5   &  1&1\\
    25   &  2&0.0692\\
    121   &3&0.0053\\
    \hline
  \end{tabularx}
  \label{table_vicsek}
\end{table}
\vspace{1cm}
\subsection{\textbf{Effectiveness with disconnected
    graphs}} \label{connectivity}

In this section, we consider MAS communicating through disconnected
graphs. In both of the following examples, we consider MAS
\eqref{agent-g-noise-Example} with $N=24$ agents where the agents are
subject to noise \eqref{w_i}. We consider $d=0.5$ in our protocols.
The simulation results show that our proposed protocols effectively
achieve $\delta$-level state synchronization where the agents
communicate through disconnected directed graphs.

\subsubsection{Disconnected directed graphs}

First, we consider a MAS with disconnected directed communication
graphs shown in figure \ref{dis-directed-net}. The disconnected directed graph consists of three bi-components. The simulation results
are presented in Figure \ref{figure10}. In the first three
sub-figures, we showed $\delta$-level state synchronization for the bi-components of the disconnected graph. We also showed the
convergence of $\rho_i(t)$ to constants.

\subsubsection{Disconnected undirected graphs}

Next, we consider a MAS with undirected communication graphs shown in
figure \ref{dis-undirected-net}. The simulation results are presented
in Figure \ref{figure9}. The simulation results illustrate the
effectiveness of our protocols in achieving $\delta$-level coherent
state synchronization for MAS with disconnected communication
networks.

\subsection{\textbf{Effectiveness with different types of
    communication graphs}}\label{Graph}

In this example, we illustrate that the protocol that we designed for
our MAS also achieves synchronization for different types of
communication graphs.  We consider MAS \eqref{agent-g-noise-Example}
with $N=121$ where the agents are subject to
noise \eqref{w_i}. In this example, the agents are communicating
through directed circulant graphs shown in
Figure~\ref{Circulant_Graph_UD}.  Figure
\ref{N121-Dcirculant-sin-delta05} shows the effectiveness of our
designed protocol \eqref{protocol-example} for MAS with directed
circulant communication graphs.

\subsection{\textbf{Robustness to different noise patterns}}\label{Noise}

In this example, we analyze the robustness of our protocols to
different noise patterns. We consider MAS with $N=121$ in
section \ref{Size} communicating through a directed Vicsek fractal
graph. In this example, we assume that agents are subject to
\begin{equation}\label{w_i2}
  w_i(t)=0.01it-\round(0.01it), \quad i=1,\ldots, N.
\end{equation}
Figure \ref{N121-Dvicsek-it-delta05} shows that our designed protocol is
robust even in the presence of noises with different patterns. 

\begin{figure}[th!]
  \centering \includegraphics[width=0.75\textwidth]{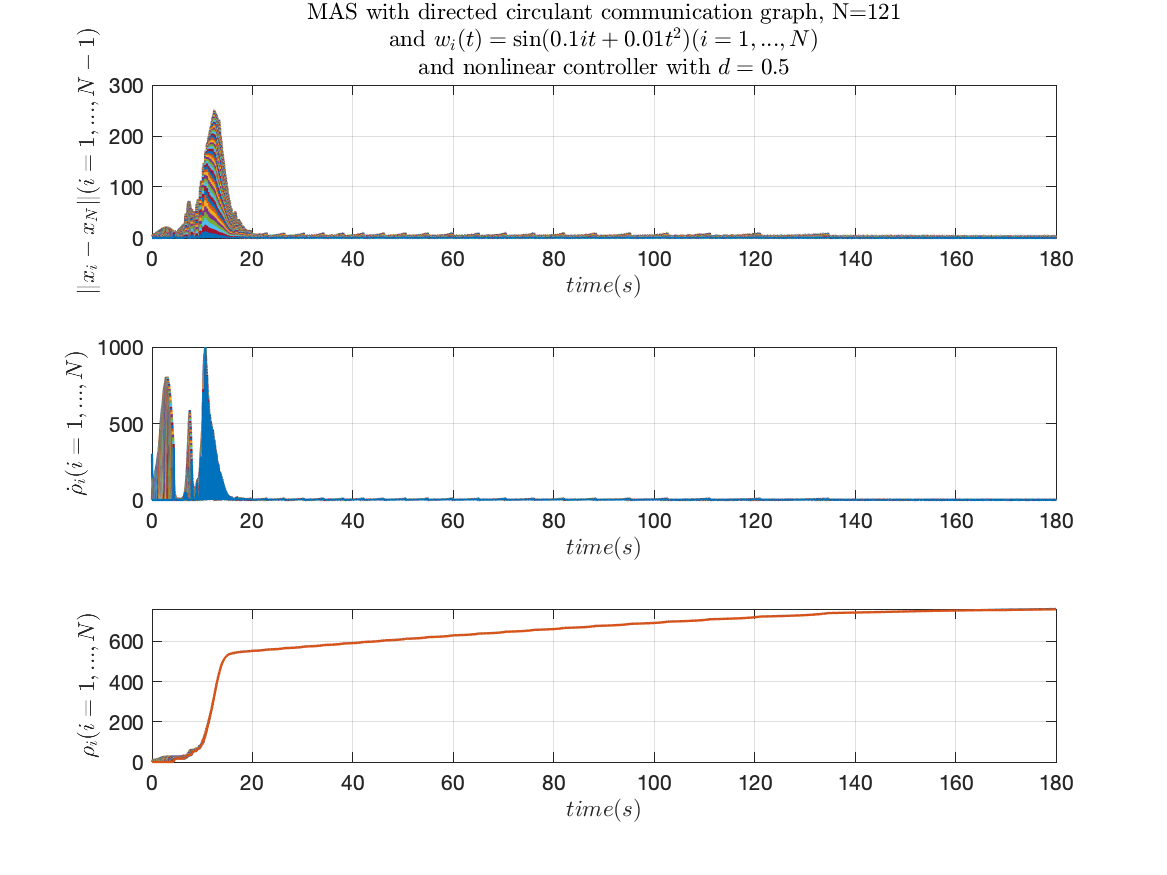}
  \vspace{-3mm}
  \caption[]{{$\delta$-level coherent state synchronization of MAS
      with $N=121$ and directed circulant communication graphs in the
      presence of disturbances via nonlinear protocol with
      $d=0.5$}} \label{N121-Dcirculant-sin-delta05}
\end{figure}
\begin{figure}[th!]
  \centering
  \includegraphics[width=0.75\textwidth]{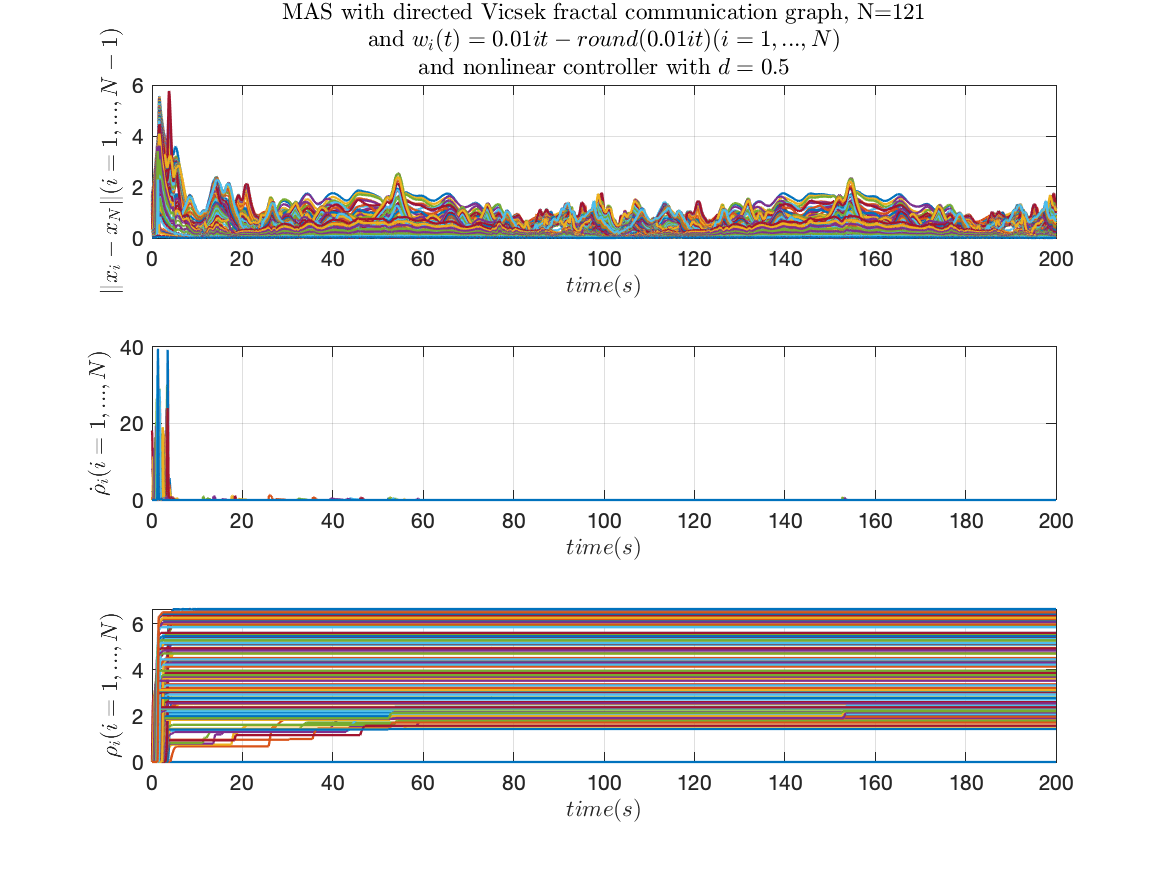}
  \vspace{-3mm}
  \caption[]{{$\delta$-level coherent state synchronization of MAS
      with $N=121$ and directed Vicsek fractal communication graphs in
      the presence of disturbances $w_i=0.01it-round(0.01it)$ via
      nonlinear protocol with
      $d=0.5$}} \label{N121-Dvicsek-it-delta05}
\end{figure}

\subsection{\textbf{Effectiveness for different values of $\delta$.}}

Finally, in this section, we show the effectiveness of the proposed
protocol for different values of $\delta$ (or, equivalently, different
values of $d$).  Similar to the previous examples, we consider the MAS
of section \ref{Size} with $N=121$ communication through directed
Vicsek fractal graphs where the agents are subject to noise
\eqref{w_i}. In this example, we choose $d=0.2$. The simulation
results presented in Figure \ref{N121-DVicsek-sin-delta02} show the
effectiveness of our protocol independent of the value of
$d$.\clearpage

\begin{figure}[th!]
  \centering
  \includegraphics[width=0.75\textwidth]{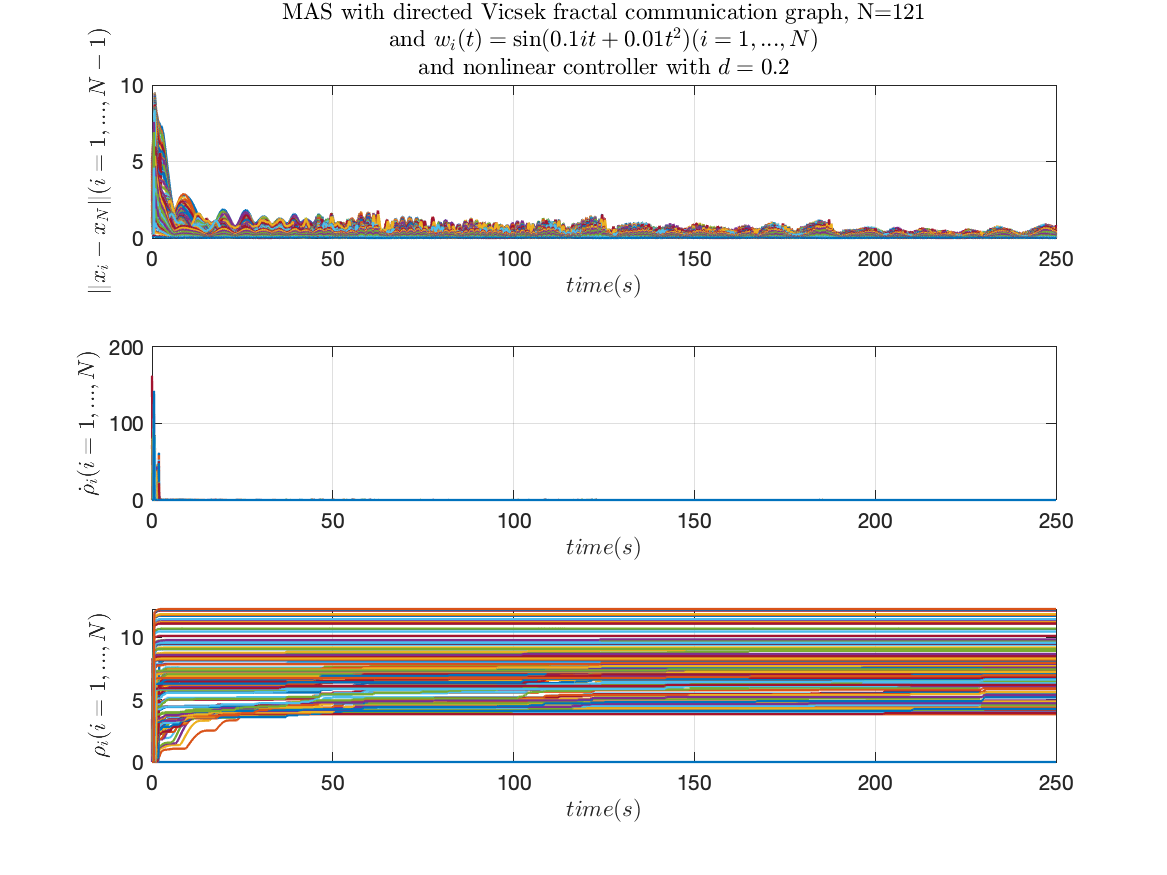}
  \vspace{-3mm}
  \caption[]{{$\delta$-level coherent state synchronization of MAS
      with $N=121$ and directed Vicsek fractal communication graphs in
      the presence of disturbances via nonlinear protocol with
      $d=0.2$}} \label{N121-DVicsek-sin-delta02}
\end{figure}

\section{Conclusion}

This paper studied scalable $\delta-$level coherent state
synchronization of MAS where agents were subject to bounded
disturbances/noises. The results of this paper can be utilized in the
string stability of vehicular platoon control and power systems. The
directions for future research on scalable $\delta-$level coherent
synchronization are $1.$ Scalable $\delta-$level coherent
synchronization of MAS with partial-state coupling; $2.$ Considering
non input additive disturbance \emph{i.e.,} where
$\im E\not\subset \im B$. The authors are conducting research in both
directions.




\newcounter{equation2}
\setcounter{equation2}{\value{equation}}
\appendix
\setcounter{equation}{\value{equation2}}%

\begin{lemma}\label{2.8}
  Consider a directed graph with Laplacian matrix $L$ which is strongly
  connected. Then, there exists $\alpha_1,\ldots,\alpha_N>0$ such that
  \begin{equation}\label{Hlyap}
    H^NL +L\T H^N \geq 3\gamma L\T L, 
  \end{equation}
  for some $\gamma >0$ with $H^N$ given by \eqref{HN} for $k=N$.
\end{lemma}

\begin{proof}
  Choose a left eigenvector of the Laplacian $L$ associated with
  eigenvalue $0$
  \[
    \begin{pmatrix} \alpha_1 & \cdots & \alpha_N \end{pmatrix} L =0.
  \]
  Because the network is strongly connected, by \cite[Theorem
  4.31]{qu-book-2009} we can choose $\alpha_1,\ldots,\alpha_N>0$ and
  obtain
  \begin{equation}\label{app1}
    H^NL +L\T H^N \geq 0
  \end{equation}
  Note that $H^N L$ has the structure of a Laplacian matrix with a
  zero row sum but also a zero column sum. The latter implies that
  $L\T H^N$ also has the structure of a Laplacian matrix. The sum of
  two Laplacian matrices still have the structure of a Laplacian
  matrix. In other words, $H^NL +L\T H^N$ has the structure of a
  Laplacian matrix. Note that this is an undirected
  graph and there is an edge between nodes $i$ and $j$ if 
  there is either an edge from node $i$ to $j$ or an edge from node $j$ to $i$
  in the original graph associated with the Laplacian matrix
  $L$. Since the graph associated with the Laplacian matrix $L$ was
  strongly connected this implies that the graph associated with the
  Laplacian matrix $H^NL +L\T H^N$ is also strongly connected. But
  then the rank of the matrix $H^NL +L\T H^N$ is equal to $N-1$ and
  hence we obtain
  \begin{equation}\label{app2}
    \text{Ker} (H^NL +L\T H^N ) = \text{Ker} L\T L = \Span \{ \textbf{1} \}.
  \end{equation}   
  \eqref{app1} and \eqref{app2} together imply that there exists a
  $\gamma$ such that \eqref{Hlyap} is satisfied.
\end{proof}

We also need to prove that the Lyapunov function in our paper is
decreasing in $\rho$ which is established in the following lemma.

\begin{lemma}\label{2.9}
  The quadratic form
  \[
    V=z\T Q_{\rho} z
  \]
  with $Q_{\rho}$ given by \eqref{Qrho} is decreasing in $\rho_i$ for
  $i=1,\ldots, k$.
\end{lemma}

\begin{proof}
  Note that
  \begin{equation}\label{appstar1}
    Q_{\rho} \textbf{1} =0,
  \end{equation}
  since $\textbf{h}_N\T \rho^{-N} \textbf{1} = \mu_N^{-1}$. We define
  \[
    z=\begin{pmatrix} z_1 \\ \vdots \\ z_N \end{pmatrix},\qquad
    \bar{z} =\begin{pmatrix} z_1-z_N \\ \vdots \\
      z_{k-1}-z_N \end{pmatrix},\qquad
    \textbf{h}_{N-1} =\begin{pmatrix} \alpha_1\\ \vdots \\
      \alpha_{N-1}\end{pmatrix}
  \]
  and we find
  \[
    V =
    \bar{z}\T (\rho^{N-1})^{-1} (\rho^{N-1}H^{N-1} - \mu_N \textbf{h}_{N-1}\textbf{h}_{N-1}\T )
    (\rho^{N-1})^{-1} \bar{z},
  \]
  using \eqref{appstar1}. Some simple algebra establishes that
  \[
     (\rho^{N-1})^{-1} (\rho^{N-1}H^{N-1} - \mu_N \textbf{h}_{N-1}\textbf{h}_{N-1}\T )
    (\rho^{N-1})^{-1} = \left[ \rho^{N-1}(H^{N-1})^{-1} + \alpha_N^{-1}\rho_N
    \textbf{1}_{N-1}\textbf{1}_{N-1}\T \right]^{-1},
  \]
  and hence
  \[
    V= \bar{z}\T \left[ \rho^{N-1}(H^{N-1})^{-1} + \alpha_N^{-1}\rho_N
    \textbf{1}_{N-1}\textbf{1}_{N-1}\T \right]^{-1} \bar{z},
  \]
  which clearly establishes that $V$ is decreasing in  $\rho_i$ for
  $i=1,\ldots, k$.
\end{proof}

\begin{lemma}\label{2.10}
  Consider $\tilde{V}_j$ as defined in \eqref{tildeVj} for $j<N$ and
  in \eqref{tildeVN} for $j=N$. There exists $M$ such that for
  $j=2,\ldots ,k$ we have
  \begin{equation}\label{Mbound}
    \tilde{V}_{j-1} \leq M \tilde{V}_j
  \end{equation}
  provided $\tilde{\rho}_j \leq 2\tilde{\rho}_i$ for $i=1,\ldots,j-1$.
\end{lemma}

\begin{proof}
  First, consider the case $j=k=N$. We have \eqref{tildeL11} with
  $\tilde{L}_{11}\in \R^{N\times N}$ is obtained from $L$ by
  permutation. We have
  \[
    \tilde{L}=\begin{pmatrix}
      \tilde{L}^{k-1}_{11} & \tilde{L}^{k-1}_{12} \\
      \tilde{L}^{k-1}_{22} & \tilde{L}^{k-1}_{22}
    \end{pmatrix}
  \]
  with $\tilde{L}^{k-1}_{22}\in \R$. The property
  $\tilde{L}_{11}\textbf{1}=0$ then implies
  \[
    (\tilde{L}^{k-1}_{11})^{-1}\tilde{L}^{k-1}_{12} =-\textbf{1},
  \]
  which implies
  \[
    \check{x}^{N-1}=\begin{pmatrix} \tilde{x}_1-\tilde{x}_N \\ \vdots
      \\ \tilde{x}_{N-1}-\tilde{x}_N 
    \end{pmatrix}.
  \]
  Using the same algebra as in the proof of Lemma \ref{2.9}, we find
  \[
    \tilde{V}_N = (\check{x}^{N-1})\T \left[ (\rho^{N-1}(\tilde{H}^{N-1})^{-1}
      + \tilde{\alpha}_N^{-1}\tilde{\rho}_N \textbf{1}\textbf{1}\T
      )^{-1} \otimes I \right]\check{x}^{N-1}.
  \]
  Choose $\beta_N$ such that $\tilde{\alpha}_i \leq \beta_{N}
  \tilde{\alpha}_N$ which implies $\tilde{H}^{N-1} \leq
  \beta_N\tilde{\alpha}_N I$. Also note $\textbf{1}\textbf{1}\T \leq N
  I$. We obtain
  \[
    \tilde{V}_N \geq  (\check{x}^{N-1})\T \left[ \left(\rho^{N-1}(\tilde{H}^{N-1})^{-1}
      + 2N\beta_N \rho^{N-1}(\tilde{H}^{N-1})^{-1} 
      \right)^{-1} \otimes I \right]\check{x}^{N-1} =
    (1+N\beta_N)^{-1}\tilde{V}_{N-1},
  \]
  which implies \eqref{Mbound} for $j=N$ provided $M>1+2N\beta_N$.

  Next, we consider $j<N$. Note that \eqref{tildeL11} implies
  \[
    \tilde{L}_{11}^j = \begin{pmatrix}
      \tilde{L}_{11}^{j-1}   & \tilde{L}_{12,a}^{j-1} \\
      \tilde{L}_{21,a}^{j-1} & \tilde{L}_{22,a}^{j-1} 
    \end{pmatrix},\qquad  \check{x}^j = \begin{pmatrix} \check{x}^j_{j-1} \\
      \check{x}^j_j \end{pmatrix},
  \]
  with $\tilde{L}_{12,a}^{j-1}\in \R^j$ and $\check{x}^j_j\in \R^n$. We have
  \begin{align*}
    [\tilde{L}_{11}^j \otimes I]\check{x}^j = \left[ \begin{pmatrix} \tilde{L}_{11}^j &
        \tilde{L}_{12}^j \end{pmatrix} \otimes I \right] \tilde{x}^k
    [\tilde{L}_{11}^{j-1} \otimes I]\check{x}^{j-1} = \left[ \begin{pmatrix} \tilde{L}_{11}^{j-1} &
        \tilde{L}_{12}^{j-1} \end{pmatrix} \otimes I \right] \tilde{x}^k.
  \end{align*}
  Since
  $\begin{pmatrix} \tilde{L}_{11}^{j-1} & \tilde{L}_{12}^{j-1} \end{pmatrix}$
  are the first $j-1$ rows of
  $\begin{pmatrix} \tilde{L}_{11}^j & \tilde{L}_{12}^j \end{pmatrix}$
  and using our decomposition of $\tilde{L}_{11}^j$ we obtain
  \[
    (\tilde{L}_{11}^{j-1}\otimes I) \check{x}^j_{j-1} + (\tilde{L}_{12,a}^{j-1}
    \otimes I) \check{x}^j_j = (\tilde{L}_{11}^{j-1}\otimes I) \check{x}^{j-1},
  \]
  which implies
  \[
    \check{x}^{j-1} = \check{x}^j_{j-1} + (S\otimes I) \check{x}^j_j,
  \]
  where $S\in \R^{j-1}$ given by
  $S=(\tilde{L}_{11}^{j-1})^{-1}\tilde{L}_{12,a}^{j-1}$. Note that
  $S\T S$ is therefore a scalar. Choose $\beta_j$ such that
  $\tilde{\alpha}_i \leq \beta_{j} \tilde{\alpha}_j$ which implies
  $\tilde{H}^{j-1} \leq \beta_j\tilde{\alpha}_j I$.
  
  We find
  \begin{align*}
    \tilde{V}_{j-1}= (\check{x}^j)\T \left[ \tilde{H}^j \tilde{\rho}^{-j} \otimes P
    \right] \check{x}^j &=
    (\check{x}^j_{j-1} + (S\otimes I) \check{x}^j_j)\T \left[
      \tilde{H}^{j-1} (\tilde{\rho}^j)^{-1}  \otimes P
    \right] (\check{x}^j_{j-1} + (S\otimes I) \check{x}^j_j)\\
    &\leq  (\check{x}^j_{j-1} )\T \left[
      \tilde{H}^{j-1} (\tilde{\rho}^j)^{-1}  \otimes P
    \right] \check{x}^j_{j-1} +  (\check{x}^j_j)\T (S\T \otimes I) \left[
      \tilde{H}^{j-1} (\tilde{\rho}^j)^{-1}  \otimes P
    \right] (S\otimes I) \check{x}^j_j \\
    &\leq \tilde{V}_j +)\check{x}^j_j)\T \left[
      S\T \tilde{H}^{j-1} (\tilde{\rho}^j)^{-1}S  \otimes P
    \right] \check{x}^j_j \\
    &\leq \tilde{V}_j +\check{x}^j_j)\T \left[
      \beta_j \alpha_j (\tilde{\rho}^j)^{-1}S\T S  \otimes P
    \right] \check{x}^j_j \\
    &\leq \tilde{V}_j + S\T S \beta_j \tilde{V}_j,
  \end{align*}
  which implies \eqref{Mbound} for $i=j$ provided $M>1+S\T S \beta_j$.
\end{proof}

\end{document}